\documentclass[journal]{IEEEtran}
\usepackage[T1]{fontenc}
\usepackage[latin9]{inputenc}
\usepackage{geometry}
\usepackage{float}
\usepackage{amsthm}
\usepackage{amsmath}
\usepackage{amssymb}
\usepackage{graphicx}
\usepackage{epsfig}
\usepackage{amsfonts}
\usepackage{amsthm}
\usepackage{cite}
\usepackage{float}
\usepackage{ifthen}
\usepackage{color}
\usepackage{graphicx}
\usepackage{cite}
\usepackage{float}
\usepackage{enumerate}
\usepackage{algorithm}
\usepackage[noend]{algpseudocode}
\usepackage{ifthen}
\usepackage{lipsum}
\newboolean{single}
\makeatletter

\theoremstyle{plain}
\newtheorem{thm}{\protect\theoremname}
\theoremstyle{plain}

\theoremstyle{plain}
\newtheorem{prop}[thm]{\protect\propositionname}

\theoremstyle{plain}

\usepackage{epstopdf}

\usepackage{units}
\providecommand{\lemmaname}{Lemma}
\providecommand{\propositionname}{Proposition}
\providecommand{\theoremname}{Theorem}
\providecommand{\lemmaname}{Lemma}
\providecommand{\propositionname}{Proposition}
\providecommand{\theoremname}{Theorem}
\makeatother
\providecommand{\lemmaname}{Lemma}
\providecommand{\propositionname}{Proposition}
\providecommand{\theoremname}{Theorem}
\begin{document}
\setboolean{single}{false}

\title{ Cross-layer Chase Combining with Selective Retransmission,  Analysis and Throughput Optimization  
for OFDM Systems}
\author{Taniya Shafique,~Zia~Muhammad~\IEEEmembership{Member,~IEEE,} and~Huy-Dung Han~\IEEEmembership{Member,~IEEE}}
%
\maketitle
\begin{abstract}
In this paper, we present bandwidth efficient retransmission method employong selective retransmission approach at modulation layer under orthogonal frequency division multiplexing (OFDM) signaling. Our proposed cross-layer design embeds a selective retransmission sub-layer in physical layer (PHY)  that targets retransmission of information symbols transmitted over poor quality OFDM sub-carriers. Most of the times, few errors in decoded bit stream result in packet failure at medium access control (MAC) layer. The unnecessary retransmission of good quality information symbols of a failed packet has detrimental effect on overall throughput of  transceiver. We propose a cross-layer Chase combining with selective retransmission (CCSR) method by blending Chase combining at MAC layer and selective retransmission in PHY. The selective retransmission in PHY targets the poor quality information symbols prior to decoding, which results into lower hybrid automatic repeat reQuest (HARQ) retransmissions at MAC layer. We also present tight bit-error rate (BER) upper bound and tight throughput lower bound for CCSR method.  In order to maximize throughput of the proposed method, we formulate optimization problem  with respect to the amount of information to be retransmitted in selective retransmission. The simulation results demonstrate significant throughput gain of the  proposed CCSR method as compared to  conventional Chase combining method.
\end{abstract}

\section{Introduction}

\label{sec:introduction} 
The contemporary wireless communication
standards such as LTE-advanced \cite{LTEDOc} integrate new technologies
to meet  increasing need of high data rate. The current and future
communication systems employ multiple-input multiple-output (MIMO)
technology due to its potential to achieve higher data rate and
diversity. In order to assure error-free communication with high throughput over dynamic wireless channels,
many packet error detection and correction protocols have evolved over
time \cite{ShuLin}. The automatic repeat reQuest (ARQ) methods combats packet loss that occurs
due to channel fading of wireless networks and achieves error-free data transfer using cyclic redundancy check (CRC) approach.
The concept of HARQ integrates ARQ and forward error correction (FEC)
codes  to provide effective means of enhancing
overall throughput of communication systems \cite{ShuLin,HARQBroadBand}. In the event of packet failure, an advanced form
of HARQ incorporates joint decoding by combining soft information
from multiple transmissions of a failed packet. Thus, HARQ is one of
the most important technologies adopted  in the latest communication
standards such as high-speed down link packet access (HSDPA), universal
mobile telecommunications system (UMTS) that pervade 3G and 4G wireless
networks to ensure data reliability. 

In type-I HARQ, the receiver requests retransmission of an erroneous packet and discards
observation of the failed packet. Type-II HARQ is  most commonly used
method and achieves higher throughput. The type-II HARQ is divided
into Chase combining HARQ (CC-HARQ) and incremental redundancy HARQ
(IR-HARQ). In CC-HARQ, the receiver preserves observations of the failed
packet and requests retransmission of full packet. The Chase receiver 
combines \cite{ChaseComb} observations of the failed packet and retransmitted
packet for joint decoding. In the event of packet failure under IR-HARQ,
the receiver requests retransmission of additional parity bits to recover
from errors. In response to the retransmission request, transmitter sends
more parity bits lowering the code rate of FEC code. After receiving requested parity bits, the receiver combines
new parity bits with buffered observations for FEC joint decoding.

Most of the research conducted on HARQ focuses on ARQ and FEC \cite{HARQBroadBand,ShenFitz}
without exploring the modulation layer. Throughput of capacity achieving FEC codes such as low density parity check (LDPC) codes and 
turbo codes is optimized for Rayleigh fading channel  in \cite{ARQJindal} with ARQ and HARQ protocols. 
Mutual information based performance analysis of HARQ over Rayleigh fading channel is provided in 
\cite{JindalHARQAnalyRayFad}.  Optimal power allocation for Chase combining based HARQ is 
optimized in \cite{OptPowerLarson,OutageLarson, AdapPowerLarson,EngAwareShami,GreenHARQ}. Without exploiting channel state
information and frequency diversity of the frequency selective channel,
partial retransmission of the original symbol stream of a failed packet
is addressed in \cite{ZiaDingGLOBECOM2008,ZhiDingPiggy,ZiaDingHARQOSTBC}.
These methods retransmit punctured packet in predetermined fashion without using channel knowledge.
Furthermore, the complexity of joint detection for \emph{partial }retransmission
is not tractable\cite{ZhiDingPiggy,ZiaDingGLOBECOM2008,Jermey2007}.
Partial retransmission of orthogonal space-time block (OSTB) coded
\cite{Alamouti} OFDM signaling is proposed in \cite{ZiaDingHARQOSTBC}.
In \cite{ARQMinitor}, for conventional ARQ protocol, full packet
retransmission at modulation layer is employed when channel gain is
below a threshold value without buffering observations of low signal-to-noise ratio (SNR) channel realization. 

In a typical failed packet, there are small number of
corrupted bits and retransmission of full packet is not necessary.
The receiver can recover from errors by retransmission of potentially
culprit bits. The OFDM signaling allows to identify poor quality bits
corresponding to the sub-carriers that have low SNR. Selective retransmission at modulation layer of OFDM signaling
proposed in \cite{ZiaSelHARQ,ETRIMIMOCondNum,IETOFDMPartial,SelecCCforOSTBC} 
achieves throughput gain as compared
to conventional HARQ methods. However, throughput optimization of selective
retransmission and performance analysis is not addressed in \cite{ZiaSelHARQ}.

 In LTE, two-levels packet retransmission achieves significant throughput gain and reduction in latency of the system. In the event of CRC failure, MAC sub-layer of user plane initiates retransmission request which results  into low latency and higher throughput. The radio link control (RLC) sub-layer combats residual packet errors by ARQ retransmission \cite{EricsonLTELLDesign, LTEBookStafania}.  These retransmission schemes do not exploit channel state information (CSI).
 Most of the contemporary communication standards such as 3G and 4G network adopt OFDM modulation due to inherent robustness to combat multi-path effect  of wireless channel and low complexity transceiver design \cite{Larsson}.  In OFDM based systems, information symbols corresponding to the different coherence bandwidth encounter different channel gains. 
The motivation of selective retransmission owing to the fact that in the event of failed 
packet under OFDM signaling at MAC layer, often receiver can recover from error(s) by retransmitting 
\emph{partial} information  corresponding to the poor quality sub-carriers. An OFDM signaling allows 
 selective retransmission of information symbols transmitted over poor quality sub-carrier at PHY level. 
After receiving the copy of information symbols corresponding to the poor quality sub-carriers, receiver jointly 
decodes data in Chase combining fashion. In this work, we propose a low complexity and bandwidth efficient 
 CCSR cross-layer design at modulation layer for OFDM signaling. We also 
provide  BER and throughput analysis in terms of tight upper BER bound and lower throughput 
bound, respectively, for the proposed retransmission scheme. The amount of information to be retransmitted for 
each sub-carrier in the event of failed packet is a function of signal-to-noise ratio (SNR) of the corresponding 
sub-carrier.  In order to maximize throughput, we use norm of  channel gain for each sub-carriers as channel quality measure and optimize 
threshold $\tau$ on channel norm for selective retransmission. The simulation 
results demonstrate that the proposed method offers substantial throughput gain as compared to the conventional CC method 
in low SNR regime. The results of proposed method show that there is marginal gap between analytical bounds and simulation 
results (Monte Carlo method) for both BER and throughput. The simulation results reveal that throughput 
gain of the proposed scheme also hold with LDPC FEC code.

We organize this manuscript as follows. First, we present the system model in Section~\ref{sec:SystemModel} and problem formulation of CCSR method for OFDM system in Section~\ref{subProblemFormulation}.
In Section~\ref{sec:perfomnceanlysis}, we present BER analysis of
  the CCSR method in terms of BER upper bound.  
Throughput analysis for the proposed CCSR is presented in
Section~\ref{sec:throughAnalysis}. Throughput optimization is performed
in Section~\ref{sec:OptThrough}. We discuss the results in Section~\ref{Sec:simulation}.
Finally, we conclude the proposed work in Section~\ref{Sec:CONCLUSION}.

\section{SYSTEM MODEL }
\label{sec:SystemModel}
The   system model under consideration employs three levels of retransmission as depicted in Figure~\ref{fig:SystemModelSCC}. The two-layer ARQ approach in LTE achieves low latency and high throughput \cite{LTEDOc}. The system model in Figure~\ref{fig:SystemModelSCC} embeds an additional retransmission sub-layer in PHY for selective retransmission under OFDM modulation with $N_s$ sub-carriers over frequency selective channel of $L$ coefficients.   An OFDM signaling converts frequency selective channel 
$\mathbf{h}$  into $N_{s}$ parallel flat-fading channels\cite{Larsson}.
The elements of a channel gain vector  $H=\big[H(1)\,\, s H(\ell)\,\,\, H(N_s)\big]^T$,
where channel vector $H$ is generated by applying
Fourier transformation matrix $F\in\mathcal{C}^{N_{s}\times N_{s}}$
on frequency selective channel $\mathbf{h}$, are independent and identically distributed (i.i.d.) along time with distribution $\mathcal{N}(0,1)$
\cite{Larsson}. The matrix model of the received vector $\mathbf{y}$ over
 $N_s$ sub-carriers can be written as 
\begin{align}
\mathbf{y}=\text{diag}(H)   \mathbf{s}+\mathbf{w},
\end{align}
where vector $\mathbf{w}\sim\mathcal{N}(0,N_{0}I)$ is an additive white
Gaussian noise vector. 
\begin{figure}[t]
\ifthenelse{\boolean{single}}{\centering \includegraphics[width=8cm]{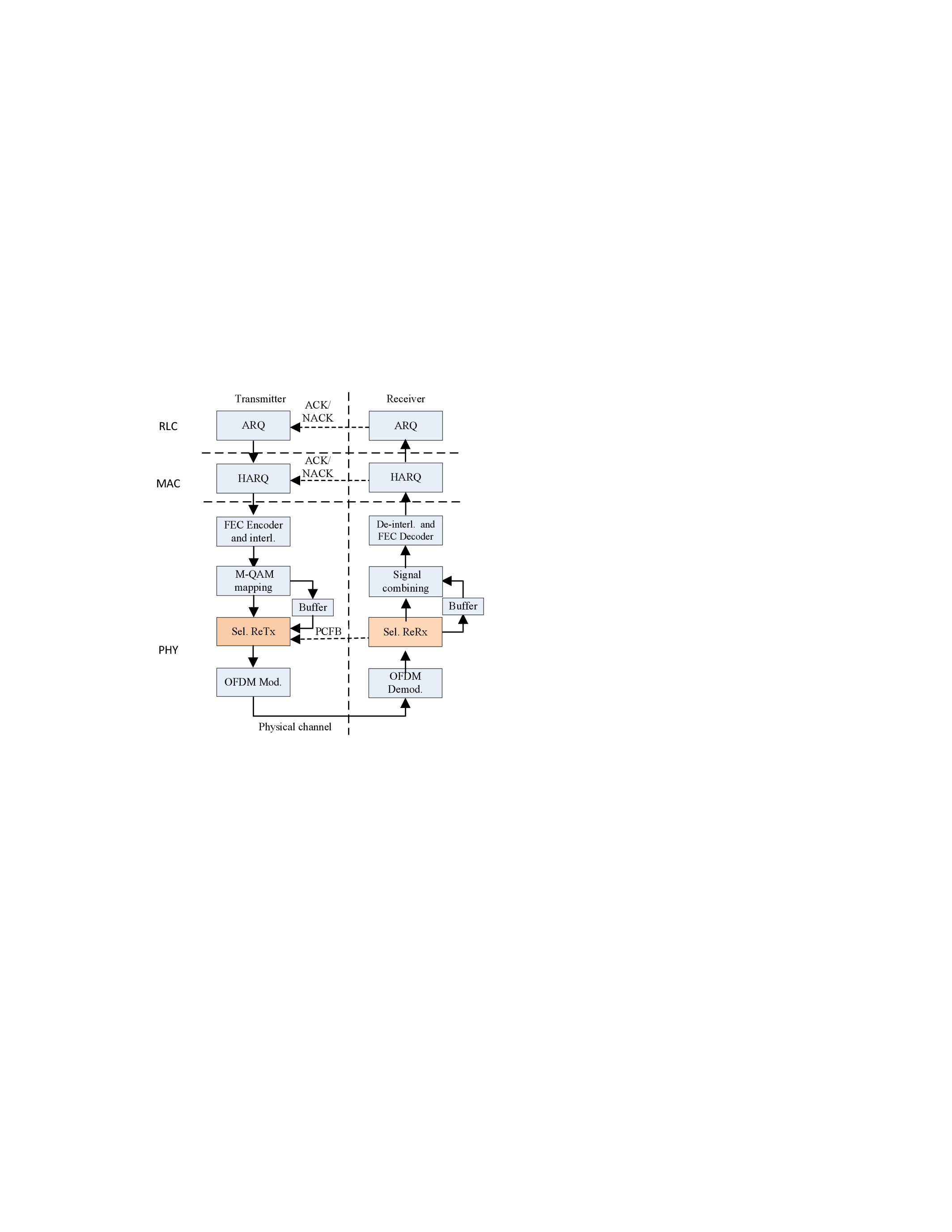} }
{\centering \includegraphics[width=8cm]{SelectiveReTxModel.eps} } 
\caption{ Cross-layer system model for Chase combining with selective retransmission for OFDM system at PHY layer}. 
\label{fig:SystemModelSCC}
\end{figure}
A typical failed packet has few erroneous bits. If we can identify unreliable bits, then full packet retransmission is unnecessary to recover failed packet. In OFDM modulation,  information bits transmitted over sub-carriers with small channel norm
$\Vert H(\ell)\Vert^{2}$ are more susceptible to the channel impairments. Thus, 
an OFDM signaling allows retransmission of targeted information symbols corresponding to the
poor quality sub-carriers instead of unnecessary retransmission of
full packet \cite{ZiaSelHARQ}. As shown in Figure~\ref{fig:SystemModelSCC}, transmitter preserves information symbol vector $\mathbf{s}=[s(1)\,\,\hdots\,\,\,s(N_s)]^T$ and transmits OFDM modulated signal. The selective retransmission module of the receiver requests retransmission of information transmitted over the poor quality  sub-carriers prior to decoding through partial channel feedback  (PCFB). The norm of gain of a sub-carrier is measure of SNR of the sub-carriers. The receiver marks information symbols for selective retransmission corresponding to the sub-carriers, which have  norm of gain below threshold $\tau$.  The threshold $\tau$ controls amount of information to be retransmitted discussed in Section~\ref{sec:perfomnceanlysis}.  In response to the selective retransmission in PHY, peer selective retransmission module of the transmitter appends requested information symbols  to the next OFDM symbol vector. Thus, each OFDM symbol vector consists of new information symbols and information symbols from the buffer in response to the selective retransmission request. The receiver then  performs joint detection by combing observation of the first transmission and subsequent selective retransmission to enhance log-likelihood ratio (LLR) of bits for FEC decoding.  Partial retransmission at modulation layer by targeting poor quality observations selectively improves BER and consequently lowers average number of retransmissions at MAC layer. HARQ layer delivers successfully decoded data units to the ARQ layer. When timeout for missing data unit occurs, ARQ layer request retransmission of  corresponding packet from the peer ARQ layer of the transmitter. We propose CCSR  selective retransmission method that achieves significant throughput gain as compared to the conventional CC-HARQ. 

Note that retransmission of more information does not increase
throughput linearly. The threshold parameter $\tau$ on the channel norm $\Vert H(\ell)\Vert^{2}$
of the $\ell$-th sub-carriers controls amount of information to be retransmitted in selective retransmission with objective to maximize throughput of the communication system. We optimize  threshold $\tau$
in order to maximize throughput $\eta$  of the transceiver under selective retransmission. Throughput of selective retransmission is function of probability
of error, which in fact is function of $\tau$. 
\section{Problem Formulation}
\label{subProblemFormulation}
Now we present   proposed cross-layer CCSR method  for OFDM 
signaling. Similar to conventional HARQ, in  CCSR method, MAC layer initiates 
retransmission in the event of CRC failure. The additional selective retransmission sub-layer in PHY layer initiates selective retransmission of information symbols transmitted over poor quality OFDM sub-carriers prior to decoding. 
Note that OFDM signaling allows retransmission of information symbols transmitted over poor 
quality sub-carriers ($\Vert {H(\ell)}\Vert^2< \tau$ ) selectively avoiding overhead of retransmission 
of information symbols corresponding to good quality sub-carriers, where $\tau$ is threshold on 
channel norm of a sub-carrier.  The receiver feeds back the partial 
channel state information (PCSI) when  each coherent time is elapsed. We assume that due to 
longer retransmission delay, each retransmission encounters independent channel. Next, we present CCSR method under OFDM signaling.

 \begin{figure}[t]
\ifthenelse{\boolean{single}}{\centering \includegraphics[width=13cm]{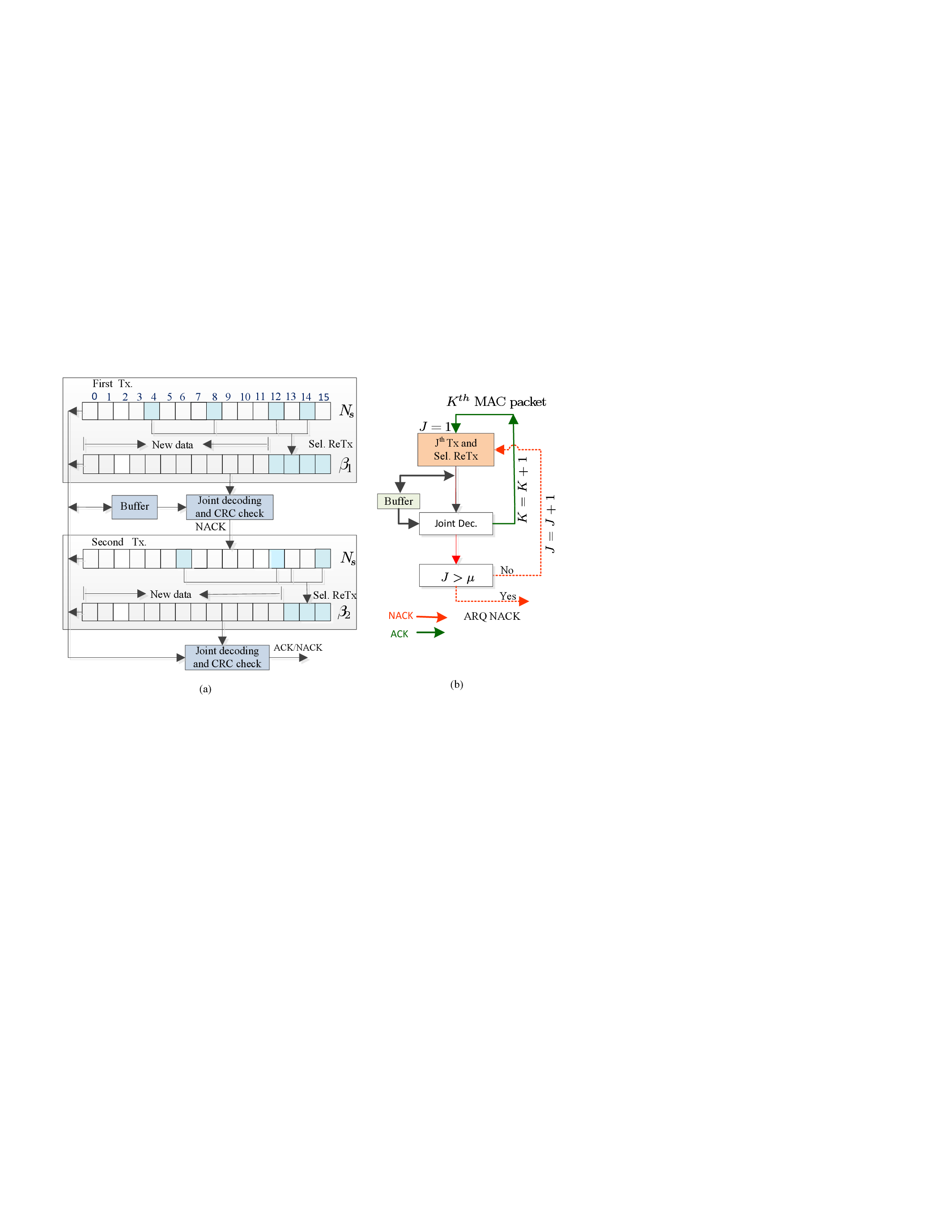}}
{\centering \includegraphics[width=8cm]{CCSRZia.eps}} 
\caption{ (a) Chase combining with selective retransmission under OFDM signaling for two transmission rounds. (b) Flow graph of CCSR method for $\mu$ transmission rounds. } 
\label{fig:CCSRMethod}
\end{figure}


The  proposed CCSR method is depicted for $\mu=2$ transmission rounds in  Figure~\ref{fig:CCSRMethod}(a). Similar to conventional CC-HARQ method,   MAC layer initiates full retransmission in the event of CRC failure. The proposed CCSR method is different from conventional CC-HARQ method in the sense that CCSR method employs an additional selective retransmission of information symbol corresponding to the  poor quality sub-carriers at PHY level for each transmission at MAC layer.
 For the first transmission of each MAC packet, proposed selective retransmission sub-layer initiates retransmission of the information symbols corresponding to the $\beta_1$ many sub-carriers which have $\Vert H_1(\ell)\Vert^2 < \tau$, where $H_1(\ell)$ represents gain of the $\ell$-th sub-carrier corresponding to full transmission of the first round at MAC layer. For example, upper dotted rectangle of Figure~\ref{fig:CCSRMethod}(a) shows that OFDM sub-carriers 4, 8, 12 and 14 ($\beta_1=4$) have $\Vert H_1(\ell)\Vert^2 < \tau$ and are marked for retransmission, where $\ell$ is index to the sub-carrier. Selective retransmission sub-layer of the transmitter in response to the selective retransmission  through feedback channel appends requested information symbols to the very next OFDM symbol.  Note that for  a very poor channel realizations, the proposed selective retransmission sub-layer in PHY  may request retransmission of all information symbols ($\beta_1=N_s$). Similarly, for good quality channel realizations, selective retransmission sub-layer can omit retransmission ($\beta_1=0$). On arrival of requested selective information, receiver performs joint detection and buffers $\beta_1+N_s$  
observations of first transmission and selective retransmission. Note that both transmitter and receiver keep track of information symbols which has been considered for selective retransmission in one transmission round. Each information symbol in one transmission round is consider only once for selective retransmission. With maximum transmission rounds $\mu$ at MAC layer, an information symbol can be considered at most $\mu$ times for selective retransmission at PHY level.

Let $H_{1s}(\ell)$ be the gain of the $\ell$-th sub-carrier corresponding to the selective retransmission for the first transmission round, where subscript "s" stands for selective retransmission. Then the combined channel response $\mathcal{H}_1(\ell)=\big[H_1(\ell)\,\, H_{1s}(\ell)\big]^{T}$ for $\beta_1$ many sub-carriers
is constructed by stacking channel of the first transmission and selective
retransmission. If there are $\beta_1$  many sub-carriers with $\Vert H_1(\ell)\Vert^2< \tau$, then there will be joint detection for $\beta_1$ sub-carrier for the first round of CCSR method. 
The estimate of $N_s-\beta_1$  information symbols which have sub-carrier gain $\Vert H_1(\ell)\Vert^2  \ge \tau$  after equalization is 
\begin{align}
\hat{s}(\ell)=s(\ell)+\frac{H_1^{*}(\ell)\mathbf{w}(\ell)}{\Vert H_1(\ell)\Vert^{2}}=s(\ell)+\mathbf{u}(\ell),
\label{eq:ZFSingleTx}
\end{align}
where $\mathbf{u}(\ell)$ is the effective noise with distribution
$\mathcal{N}(0,\Vert H(\ell)\Vert ^{-2}N_{0})$.  Also the estimate of $\beta_1$ information symbols corresponding to the poor quality sub-carriers from the first full transmission as a result of joint
detection is 
\begin{align}
\hat{s}(\ell)= & s(\ell)+\underbrace{\Vert\mathcal{H}_1(\ell)\Vert^{-2}\mathcal{H}_1^{H}(\ell)\tilde{\mathbf{w}}(\ell)}_{\tilde{\mathbf{u}}(\ell)},\label{eq:SelDetector}
\end{align}
where $\tilde{\mathbf{u}}(\ell)\thicksim\mathcal{N}(0,\Vert\mathcal{H}_1(\ell)\Vert^{-2}N_{0}I_2)$
and $\tilde{\mathbf{w}}(\ell)\thicksim\mathcal{N}(0,\, N_{0}I_{2})$.
Note that MAC layer is unaware of selective retransmission sub-layer in  PHY. The selective retransmission followed by joint detection at modulation layer enhances reliability of the  decoded bits resulting into few CRC failure at MAC layer by selectively retransmitting poor quality information symbols.  

In the event of CRC failure, MAC layer initiates next round of HARQ retransmission by sending NACK signal to the peer MAC layer for the full retransmission of failed data as shown in the lower dotted rectangle of the Figure~\ref{fig:CCSRMethod}(a). In  response to NACK from MAC layer, transmitter retransmits failed full packet similar to conventional CC-HARQ as shown in Figure~\ref{fig:CCSRMethod}(a). Similar to the first transmission, selective retransmission sub-layer initiates selective retransmission of poor quality symbols of retransmitted full packet from MAC layer. 
As a result of selective retransmission, transmitter appends  $\beta_2$ many ($\beta_2=3$ in Figure~\ref{fig:CCSRMethod}(a) )  information symbols to the next immediate OFDM symbol.
When selective retransmission is employed in PHY,  Chase combining  processes $2N_s+\beta_1+\beta_2$ observations instead of  $2N_s$  observations for joint detection. Note that   $E\left[\beta_1\right ]=E\left[\beta_2\right] =N_s P(\Vert H_1(\ell)\Vert^2 < \tau)=N_s P(\Vert H_{1s}(\ell)\Vert^2 < \tau)=N_s m $, where $H_{1s}(\ell)$ is channel gain of the 
$\ell$-th sub-carrier during selective retransmission.

   Let $\mu$ be the maximum number of allowed retransmission of a MAC packet and $J$ be the  round counter for the transmission of the $k$-th MAC packet of a HARQ process for  CCSR method. At the end of the $J$-th round, where $J=1,\hdots,\mu$, receiver
combines buffered observations of full transmissions and selective
retransmissions of $J$ rounds for joint detection. If CRC fails after
$\mu$ rounds of MAC layer, receiver clears observations from the
buffer and sends signal to the ARQ layer. In response to the
packet failure, ARQ layer initiates new retransmission round by initiating
retransmission request to the peer ARQ layer. The flow graph of CCSR  protocol for the $k$-th MAC packet is presented in Figure~\ref{fig:CCSRMethod}(b) and described in Algorithm~\ref{AlgoCCSRRev}. If  CRC failure occurs after $\mu$ MAC  retransmissions, retransmission of failed packet is initiated from ARQ layer. The algorithm of CCSR method is described as follows:
\begin{algorithm}
\caption{CCSR protocol}
\label{AlgoCCSRRev}
\begin{algorithmic}[1]
\State   $J=1$ corresponds to the first transmission of the $k$-th MAC  packet
\State  Selective retransmission of $\beta_J$ sub-carriers and buffering of $N_s$ observations
\State   Joint decoding from $JN_s+\sum_{i=1}^{J} \beta_i$ observations and  CRC check
          \If {CRC satisfies} $k=k+1$, discard observations and go to 1
					\EndIf 
   \If {$J\ge \mu$} declare packet  loss, discard observation, ARQ sub-layer initiates NACK for retransmission of packet and go to 1 \EndIf
\State   $J=J+1$   and go to 2 
\end{algorithmic}
\end{algorithm}

\section{Performance Analysis}
\label{sec:perfomnceanlysis} In this section, we present tight upper bounds on BER 
of joint detection after signal combining and prior to channel decoding for cross-layer CCSR method. One round of transmission of CCSR method includes one full transmission of MAC packet followed by selective retransmission of information symbols corresponding to the sub-carriers which have $\Vert H(\ell) \Vert^2< \tau$, where $H(\ell)$ is gain of $\ell$-th sub-carrier.  We assume that full transmission by MAC layer and subsequent  selective retransmission by PHY layer encounter independent channel realizations.  In this analysis, we consider maximum of $\mu$ transmission rounds at MAC layer. Similar to conventional  Chase combining, joint detection for the $J$-th round combines observations buffered up to  $J$ rounds. Thus, probability of error $P_{e_J}$ of the joint detection of the $J$-th round is lower than that of $J-1$-th round ($P_{e_J} < P_{e_{J-1}}$), where $J\le \mu  \in I^{+}$. In order to evaluate throughput of the proposed CCSR method in Section~\ref{sec:throughAnalysis}, we derive closed form expression for the upper bound on BER of joint detection of $J$-th round under maximum number of $\mu$ transmission rounds of a MAC packet. Now we evaluate $P_{e_J}$ for the $J$-th round.

\subsection{BER analysis of the first round of CCSR}
\label{subsec:BERRound1}
Let $H_1$ and $H_{1s}$ be the complex gain vectors of length $N_s$  corresponding to the full transmission by MAC layer  
and subsequent selective retransmission from PHY layer, respectively. Each element $H_1(\ell)$ and $H_{1s}(\ell)$ of the complex channel gain vectors of OFDM sub-carriers follows Gaussian distribution with  zero-mean and unit variance \cite{Larsson}.  One MAC data unit is mapped to $N_s$ information symbols using M-QAM modulation, where $N_s$ is the number of sub-carriers of an OFDM symbol. Prior to decoding, selective retransmission sub-layer initiates selective retransmission of $\beta_1$ poor quality information symbols as shown in upper dotted  rectangle of  Figure~\ref{fig:CCSRMethod}(a). We denote the outcomes $\Vert H_1(\ell)\Vert^{2}\ge\tau$ and $\Vert H_1(\ell)\Vert^{2}<\tau$ of the first full transmission
by the events $\xi$ and $\xi^{c}$, respectively. The probabilities of events $\xi$ and $\xi^{c}$ are
$P(\xi^{c})  =P(\Vert H_1(\ell)\Vert^{2}<\tau)\;\mbox{and}\; P(\xi)=P(\Vert H_1(\ell)\Vert^{2}\ge\tau),$
respectively, where random variable $\chi_{1}=\Vert H_1(\ell)\Vert^{2}$ has chi-square
distribution of degree $2$ \cite{DgitalCommProakis} and $P(\xi^{c})=1-P(\xi)=P(\chi_{1}<\tau)$
. For Rayleigh fading channel, real and imaginary components of complex channel coefficient of a sub-carrier have zero-mean and variance $\sigma^{2}=\frac{1}{2}$.
 When event $\xi$ occurs,  selective retransmission sub-layer omits retransmission of that particular sub-carriers.

When event $\xi^{c}$ occurs for $\ell$-th sub-carrier, selective retransmission sub-layer request retransmission of that very information symbol $s(\ell)$ and receiver performs joint detection by
combining observation of the full transmission and subsequent selective retransmission.
Note that the random variable $\Vert\mathcal{H}_1(\ell)\Vert^{2}$ in \eqref{eq:SelDetector}
also has chi-square distribution of degree $4$. Also that $\Vert\mathcal{H}_1(\ell)\Vert^{2}=\chi_{1}+\chi_{2}$, 
where chi-square random variables $\chi_{1}$ and $\chi_{2}=\Vert H_{1s}(\ell)\Vert^{2}$
are i.i.d. of degree $2$ each. The bit-error probability
of joint detection for selective retransmission over Rayleigh fading channel is
\begin{align}
P_{e_{1}} & =E_{H}\Big [P(\xi)\,\,\, P_{e\vert\xi}+P(\xi^{c})\,\,\, P_{e\vert\xi^{c}}\Big ],
\label{eq:AvgOverH}
\end{align}
where $P_{e\vert\xi}$ and $P_{e\vert\xi^{c}}$ are the conditional bit-error probabilities of  detection from single observation and joint detection,
respectively. 

The probability of error for joint detection of first round of CCSR is 
\ifthenelse{\boolean{single}}
{
\begin{align}
P_{e_{1}} & = P(\xi)c E_{H\vert\xi}\left[Q\left(\sqrt{g\frac{\chi_{1}}{N_{0}}}\right)\right]+ P(\xi^{c})c E_{\mathcal{H}\vert\xi^{c}}\left[Q\left(\sqrt{g\frac{\chi_{1}+\chi_{2}}{N_{0}}}\right)\right],
\label{eq:ProbOfError}
\end{align} 
}
{
\begin{align}
P_{e_{1}} & = P(\xi)c E_{H\vert\xi}\left[Q\left(\sqrt{g\frac{\chi_{1}}{N_{0}}}\right)\right]+ P(\xi^{c})c E_{\mathcal{H}\vert\xi^{c}}\nonumber\\
&\left[Q\left(\sqrt{g\frac{\chi_{1}+\chi_{2}}{N_{0}}}\right)\right],
\label{eq:ProbOfError}
\end{align}
} 
where $E_{H\vert\xi^{c}}$ and $E_{\mathcal{H}\vert\xi}$ are conditional
expectations. Also,  $c$ and $g$ are modulation constants \cite{Larsson}.
The conditional probability density function $f_{\chi_1\vert\xi^c}(x_1)$  of $f_{\chi_1}(x_1)$ when $\chi_1\ge\tau$ is $f_{\chi_1\vert\xi^c}(x_1) =\dfrac{f_{\chi_1}(x_1)}{P(\xi^c)}$. In order to solve first term of \eqref{eq:ProbOfError}, we have \cite{StarkWoods}
\begin{align}
\begin{split}
\displaystyle  E_{H\vert\xi}\left[Q\left(\sqrt{g\frac{\chi_{1}}{N_{0}}}\right)\right]=&\int_{\tau}^{\infty} Q \left(\sqrt{g\frac{x_{1}}{N_{0}}}\right) \dfrac{f_{\chi_1}(x_1)}{P(\xi)} dx_1
\end{split}
\label{eq:SCC2}
\end{align}
Similarly,
\ifthenelse{\boolean{single}}
{
\begin{align}
\begin{split}
\displaystyle  E_{H\vert\xi^{c}}\left[Q\left(\sqrt{g\frac{\chi_{1}+\chi_2}{N_{0}}}\right)\right]=
&\int_{x_1=0}^{\tau} \int_{x_2=0}^{\infty}Q\left(\sqrt{g\frac{x_{1}+x_2}{N_{0}}}\right)  \dfrac{f_{\chi_1}(x_1)f_{\chi_2}(x_2)}{P(\xi^{c})} dx_2 dx_1
\end{split}
\label{eq:SCC3}
\end{align}
 }
{
\begin{align}
\displaystyle  E_{H\vert\xi^{c}}&\left[Q\left(\sqrt{g\frac{\chi_{1}+\chi_2}{N_{0}}}\right)\right]=
\int_{x_1=0}^{\tau} \int_{x_2=0}^{\infty}Q\left(\sqrt{g\frac{x_{1}+x_2}{N_{0}}}\right) \nonumber\\
& \dfrac{f_{\chi_1}(x_1)f_{\chi_2}(x_2)}{P(\xi^{c})} dx_2 dx_1
\label{eq:SCC3}
\end{align}
} 
The upper bound on $P_{e_1}$ in \eqref{eq:ProbOfError} using approximation
of Q-function \cite{QFuncApprox}, \eqref{eq:SCC2} and \eqref{eq:SCC3} can be written as \cite{SelecCCforOSTBC,StarkWoods } 
\ifthenelse{\boolean{single}}
{
\begin{align}
\begin{split}
P_{e_1}\le&\frac{c}{12} \int_\tau^\infty \exp(-g\frac{x_1}{2N_0})  f_{\chi_1}(x_1) dx_1
+\frac{c}{4}  \int_\tau^\infty \exp(-g\frac{4 x_1}{3.2N_0})  f_{\chi_1}(x_1) dx_1\\	
&+\frac{c}{12} \int_0^\tau \exp(-g\frac{x_1}{2N_0})  f_{\chi_1}(x_1) dx_1	
  \int_0^\infty \exp(-g\frac{x_2}{2N_0})  f_{\chi_2}(x_2) dx_2\\
	&+\frac{c}{4}  \int_0^\tau \exp(-g\frac{4x_1}{3.2N_0})  f_{\chi_1}(x_1) dx_1	
	  \int_0^\infty \exp(-g\frac{4x_2}{3.2N_0})  f_{\chi_2}(x_2) dx_2,
	\label{eq:UBSCC}
\end{split}
\end{align}
 }
{
\begin{align}
&P_{e_1}\le\frac{c}{12} \int_\tau^\infty \exp(-g\frac{x_1}{2N_0})  f_{\chi_1}(x_1) dx_1+\frac{c}{4}  \int_\tau^\infty \nonumber\\
&\exp(-g\frac{4 x_1}{3.2N_0})  f_{\chi_1}(x_1) dx_1	
+\frac{c}{12} \int_0^\tau \exp(-g\frac{x_1}{2N_0})  f_{\chi_1}(x_1) \nonumber\\	
 & dx_1\int_0^\infty \exp(-g\frac{x_2}{2N_0})  f_{\chi_2}(x_2) dx_2
	+\frac{c}{4}  \int_0^\tau \exp(-g\frac{4x_1}{3.2N_0})  \nonumber\\		
	 &f_{\chi_1}(x_1) dx_1 \int_0^\infty \exp(-g\frac{4x_2}{3.2N_0})  f_{\chi_2}(x_2) dx_2,
	\label{eq:UBSCC}
\end{align} 
} 

Note that  $\rho=\sqrt{\frac{\sigma^2 N_{o}}{g\sigma^2+N_{o}}}$, $\rho_{1}=\sqrt{\frac{\sigma^{2}N_{o}}{g_{1}\sigma^{2}+N_{o}}}$
and $g_{1}=\frac{4g}{3}$. 
For 4-QAM constellation $g=2$ and $c=\frac{2}{\text{log}_2M}$.
 By simplifying \eqref{eq:UBSCC}, we have upper bound on probability of error of the joint detection of the first round as follows \cite{DgitalCommProakis}
\ifthenelse{\boolean{single}}
{
\begin{align}
\begin{split}
P_{e_{1}}  \le&\frac{c}{12}\left(\frac{\rho}{\sigma}\right)^{2}\exp\left(-\frac{\tau}{2\rho^{2}}\right)+\frac{c}{4}\left(\frac{\rho_{1}}{\sigma}\right)^{2}\exp\left(-\frac{\tau}{2\rho_{1}^{2}}\right)\\
 & +\frac{c}{12}\left(\frac{\rho}{\sigma}\right)^{4}\left(1-\exp\left(-\frac{\tau}{2\rho^{2}}\right)\right)
 +\frac{c}{4}\left(\frac{\rho_{1}}{\sigma}\right)^{4}\left(1-\exp\left(-\frac{\tau}{2\rho_{1}^{2}}\right)\right).
\end{split}
\label{eq:BERUpperBoundJointSCC}
\end{align}
 }
{
\begin{align}
P_{e_{1}}  \le&\frac{c}{12}\left(\frac{\rho}{\sigma}\right)^{2}\exp\left(-\frac{\tau}{2\rho^{2}}\right)+\frac{c}{4}\left(\frac{\rho_{1}}{\sigma}\right)^{2}\exp\left(-\frac{\tau}{2\rho_{1}^{2}}\right)\nonumber\\
 & +\frac{c}{12}\left(\frac{\rho}{\sigma}\right)^{4}\left(1-\exp\left(-\frac{\tau}{2\rho^{2}}\right)\right)
 +\nonumber\\
&\frac{c}{4}\left(\frac{\rho_{1}}{\sigma}\right)^{4}\left(1-\exp\left(-\frac{\tau}{2\rho_{1}^{2}}\right)\right).
\label{eq:BERUpperBoundJointSCC}
\end{align} 
} 
\noindent
Now we evaluate  BER upper bound of joint detection for second round under  CCSR method.
\subsection{BER analysis of second round of  CCSR}
\label{SubBERCCSR}
Let  $H_{2}(\ell)$ be the gain of the $\ell$-th
sub-carrier corresponding to the  full retransmission initiated from the  MAC layer as a result of CRC failure of MAC packet of the first round and $H_{2s}(\ell)$ be the channel gain  for selective retransmission. In a similar fashion to the first round, the receiver marks data symbol of the second round transmitted over $\ell$-th sub-carrier for selective retransmission for the second round 
 if $\Vert H_{2}\Vert ^2 < \tau$. The channel gains $H_1(\ell)$ and $H_{2}(\ell)$ of the full transmission of the first and second round, respectively, of the $\ell$-th sub-carrier are independent, which results into four possible joint channel vectors for joint detection. We denote each possible outcome of joint channel vector by an event.
The  probability of error of joint detection of the second round is 
\ifthenelse{\boolean{single}}
{
\begin{align}
\begin{split}
P_{e_{2}}  =&E_{H}\big[P(\xi_{1})P_{e\vert\xi_{1}}+P(\xi_{2})P_{e\vert\xi_{2}}+
 P(\xi_{3})P_{e\vert\xi_{3}}+P(\xi_{4})P_{e\vert\xi_{4}}\big],
\label{eq:CCSRjoitBER}
\end{split}
\end{align}
 }
{
\begin{align}
P_{e_{2}}  =&E_{H}\big[P(\xi_{1})P_{e\vert\xi_{1}}+P(\xi_{2})P_{e\vert\xi_{2}}+
 P(\xi_{3})P_{e\vert\xi_{3}}+\nonumber\\
&P(\xi_{4})P_{e\vert\xi_{4}}\big],
\label{eq:CCSRjoitBER}
\end{align}
} 
where events $\xi_{1},\;\xi_{2},\,\xi_{3}$ and $\xi_{4}$ correspond to the four joint channel vectors defined
as follows: 
\begin{enumerate}
\item Event $\xi_{1}$ occurs when $\Vert H_1(\ell)\Vert^{2}\ge \tau$ and $\Vert H_{2}(\ell)\Vert^{2}\ge \tau$ for the first and second full transmissions, respectively. The 
resulting joint channel for joint detection of CCSR is $\mathcal{H}_{1}(\ell)=[H_1(\ell)\;\; H_{2}(\ell)]^{T}$,
where $H_1(\ell)$ and $H_{2}(\ell)$ are i.i.d. channel realizations
with Gaussian distribution of zero-mean and unit variance. The channel
norm $\Vert\mathcal{H}_{1}(\ell)\Vert^{2}=\Vert H_1(\ell)\Vert^{2}+\Vert H_{2}(\ell)\Vert^{2}$
has chi-square distribution. 
\item Event $\xi_{2}$ occurs when $\Vert H_1(\ell)\Vert^{2}< \tau \mbox{ and }\Vert H_{2}(\ell)\Vert^{2}\ge \tau$.
The resulting joint channel response for joint detection of CCSR is
$\mathcal{H}_{2}(\ell)=[H_1(\ell)\;\; H_{1s}(\ell)\;\; H_{2}(\ell)]^{T}$. 
\item Event $\xi_{3}$ occurs when $\Vert H_1(\ell)\Vert^{2}\ge \tau \mbox{ and }\Vert H_{2}(\ell)\Vert^{2}< \tau$.
The resulting joint channel response for joint detection of CCSR is
$\mathcal{H}_{3}(\ell)=[H_1(\ell)\;\; H_{2}(\ell)\;\; H_{2s}(\ell)]^{T}$, where $H_{2s}(\ell)$ is channel gain
of the $\ell$-th sub-carrier selected for retransmission during selective  retransmission of the second round  of packet at MAC layer.  
\item Event $\xi_{4}$ occurs when $\Vert H_1(\ell)\Vert^{2}< \tau$ and
$\Vert H_{2}(\ell)\Vert^{2}< \tau$. The resulting joint channel for joint
detection of CCSR is $\mathcal{H}_{4}(\ell)=[H_1(\ell)\;\; H_{1s}(\ell)\;\; H_{2}(\ell)\;\; H_{2s}(\ell)]^{T}$,
where $H_{1s}(\ell)$ and $H_{2s}(\ell)$ are the channels corresponding
to the selective retransmissions of the first round and second round, respectively.   Note that random variables $\Vert H_{1s}(\ell)\Vert^{2}$
and $\Vert H_{2s}(\ell)\Vert^{2}$ are also i.i.d. with chi-square distribution
of degree  $2$ each. 
\end{enumerate}
The second and third terms in \eqref{eq:CCSRjoitBER} are equivalent due
to the fact that $P(\xi_{2})=P(\xi_{3})$ and random variables $\Vert\mathcal{H}_{2}(\ell)\Vert$ and
$\Vert\mathcal{H}_{3}(\ell)\Vert$ are i.i.d. Therefore, $E_{H}\Big[P(\xi_{2})P_{e\vert\xi_{2}}+P(\xi_{3})P_{e\vert\xi_{3}}\Big]=2E_{H}\Big[P(\xi_{2})P_{e\vert\xi_{2}}\Big]$.
Note that all channel realizations of the $\ell$-th sub-carrier of
an OFDM system are i.i.d. with Gaussian distribution of zero-mean  and
unit variance. In order to achieve  upper bound on BER for joint detection of CCSR method, we rewrite \eqref{eq:CCSRjoitBER} 
as follows:
\begin{align}
 & P_{e_{2}}=cE_{H}\Bigg[P(\xi_{1})Q\left(\sqrt{\frac{g \Vert \;\mathcal{H}_{1}(\ell)\Vert^{2}}{N_{0}}}\right)+2P(\xi_{2}) \nonumber \\
\negmedspace\!\!\!\!\!\!\!\!\!\! & Q\left(\negmedspace\negmedspace\sqrt{\frac{g\Vert \;\mathcal{H}_{2}(\ell)\Vert^{2}}{N_{0}}}\negmedspace\right)\negmedspace\negmedspace+\negmedspace P(\xi_{4})Q\left(\negmedspace\negmedspace\sqrt{\frac{g\Vert \;\mathcal{H}_{4}(\ell)\Vert^{2}}{N_{0}}}\right)\negmedspace\negmedspace\Bigg].
\label{eq:CCSRJointQfunc}
\end{align}
Note that $\Vert \mathcal{H}_{1}(\ell)\Vert^{2}=\chi_{1}+\chi_{2}$ in the first term of \eqref{eq:CCSRJointQfunc} is sum if two i.i.d. chi-square random variables , where $\chi_{1} \ge \tau$ and  $\chi_{2} \ge \tau$.  Using approximation of Q-function in \cite{QFuncApprox} and, following \eqref{eq:SCC2} and \eqref{eq:SCC3}, we have
\ifthenelse{\boolean{single}}
{
\begin{align}
\begin{split}
E_{H}\left[P(\xi_{1})P_{e}\vert\xi_{1}\right]  =  &cE_{H}\left[P(\xi_{1})Q\left(\sqrt{\frac{g\Vert\mathcal{H}_{1}(\ell)\Vert^{2}\big\vert \xi_1}{N_{0}}}\right)\right]\\
\le & \frac{c}{12}\int_{\tau}^{\infty}\exp(-g\frac{x_{1}}{2N_{0}})f_{\chi_{1}}(x_{1})dx_{1}\int_{\tau}^{\infty}\exp(-g\frac{x_{2}}{2N_{0}})f_{\chi_{2}}(x_{2})dx_{2}+\\
  &\frac{c}{4}\int_{\tau}^{\infty}\exp(-g\frac{4x_1}{3.2N_{0}})f_{X_{1}}(x_{1})dx_{1}\int_{\tau}^{\infty}\exp(-g\frac{4x_{2}}{3.2N_{0}})f_{\chi_{2}}(x_{2})dx_{2}\\
= & \frac{c}{12}\left(\frac{\rho}{\sigma}\right)^{4}\left(\exp\left(-\frac{\tau}{2\rho^{2}}\right)\right)^{2}+
 \frac{c}{4}\left(\frac{\rho_{1}}{\sigma}\right)^{4}\left(\exp\left(-\frac{\tau}{2\rho_{1}^{2}}\right)\right)^{2}.
\label{eq:CCSR1}
\end{split}
\end{align}
 }
{
\begin{align}
&E_{H}\left[P(\xi_{1})P_{e}\vert\xi_{1}\right] \! = \!\!cE_{H}\!\!\left[P(\xi_{1})Q\!\!\left(\sqrt{\frac{g\Vert\mathcal{H}_{1}(\ell)\Vert^{2}\big\vert \xi_1}{N_{0}}}\right)\right]\nonumber\\
\le & \frac{c}{12}\int_{\tau}^{\infty}\exp(-g\frac{x_{1}}{2N_{0}})f_{\chi_{1}}(x_{1})dx_{1}\int_{\tau}^{\infty}\exp(-g\frac{x_{2}}{2N_{0}})\nonumber\\
  &f_{\chi_{2}}(x_{2})dx_{2}+\frac{c}{4}\int_{\tau}^{\infty}\exp(-g\frac{4x_1}{3.2N_{0}})f_{\chi_{1}}(x_{1})dx_{1}\nonumber\\
	&\int_{\tau}^{\infty}\exp(-g\frac{4x_{2}}{3.2N_{0}})f_{\chi_{2}}(x_{2})dx_{2}\nonumber\\
= & \frac{c}{12}\left(\frac{\rho}{\sigma}\right)^{4}\!\!\left(\exp\left(-\frac{\tau}{2\rho^{2}}\right)\!\!\right)^{2}+\!\!\!\!
 \frac{c}{4}\left(\frac{\rho_{1}}{\sigma}\right)^{4}\!\!\left(\exp\left(-\frac{\tau}{2\rho_{1}^{2}}\right)\!\!\right)^{2}.
\label{eq:CCSR1}
\end{align} 
} 
Similarly, 
\ifthenelse{\boolean{single}}
{
\begin{align}
\begin{split}
E_{H}&\left[P(\xi_{2})P_{e}\vert\xi_{2}\right]  =
  cE_{H}\left[P(\xi_{2})Q\left(\sqrt{\frac{g\Vert\mathcal{H}_{2}(\ell)\Vert^{2}\big\vert \xi_2}{N_{0}}}\right)\right]\\
\le & \frac{c}{12}\int_{0}^{\tau}\exp(-g\frac{x_{1}}{2N_{0}})f_{\chi_{1}}(x_{1})dx_{1}\int_{0}^{\infty}\exp(-g\frac{x_{1s}}{2N_{0}})f_{\chi_{1s}}(x_{1s})dx_{1s}\int_{\tau}^{\infty}\exp(-g\frac{x_{2}}{2N_{0}})f_{\chi_{2}}(x_{2})dx_{2}+\\
 & \frac{c}{4}\int_{0}^{\tau}\exp(-g\frac{4x_{1}}{3.2N_{0}})f_{\chi_{1}}(x_{1})dx_{1}\int_{0}^{\infty}\exp(-g\frac{4x_{1s}}{3.2N_{0}})f_{\chi_{1s}}(x_{1s})dx_{1s}\int_{\tau}^{\infty}\exp(-g\frac{4x_{2}}{3.2N_{0}})f_{\chi_{2}}(x_{2})dx_{2},
\end{split}
\label{eq:SubEqCCSR}
\end{align}
 }
{
\begin{align}
&E_{H}\left[P(\xi_{2})P_{e}\vert\xi_{2}\right]\!=\!
  cE_{H}\!\left[P(\xi_{2})Q\!\left(\!\sqrt{\frac{g\Vert\mathcal{H}_{2}(\ell)\Vert^{2}\big\vert \xi_2}{N_{0}}}\right)\!\right]	\nonumber\\
&\le  \frac{c}{12}\int_{0}^{\tau}\exp(-g\frac{x_{1}}{2N_{0}})f_{\chi_{1}}(x_{1})dx_{1}\int_{0}^{\infty}\exp(-g\frac{x_{1s}}{2N_{0}})\nonumber\\
&f_{\chi_{1s}}(x_{1s})dx_{1s}\int_{\tau}^{\infty}\exp(-g\frac{x_{2}}{2N_{0}})f_{\chi_{2}}(x_{2})dx_{2}+\nonumber\\
 & \frac{c}{4}\int_{0}^{\tau}\exp(-g\frac{4x_{1}}{3.2N_{0}})f_{\chi_{1}}(x_{1})dx_{1}\int_{0}^{\infty}\exp(-g\frac{4x_{1s}}{3.2N_{0}})\nonumber\\
&f_{\chi_{1s}}(x_{1s})dx_{1s}\int_{\tau}^{\infty}\exp(-g\frac{4x_{2}}{3.2N_{0}})f_{\chi_{2}}(x_{2})dx_{2},
\label{eq:SubEqCCSR}
\end{align}
} 
where $\chi_{1}<\tau$,   $\chi_{1s} \in \mathcal{R}$ and  $\chi_{2} \ge \tau$. Simplifying \eqref{eq:SubEqCCSR}, we have
\ifthenelse{\boolean{single}}
{
\begin{align}
\begin{split}
E_{H}&\left[P(\xi_{2})P_{e}\vert\xi_{2}\right]  \le  \frac{c}{12}\left(\frac{\rho}{\sigma}\right)^{6}\left(1-\exp\left(-\frac{\tau}{2\rho^{2}}\right)\right)\\
&\left(\exp\left(-\frac{\tau}{2\rho^{2}}\right)\right)+\\
 & \frac{c}{4}\left(\frac{\rho_{1}}{\sigma}\right)^{6}\left(1-\exp\left(-\frac{\tau}{2\rho_{1}^{2}}\right)\right)\left(\exp\left(-\frac{\tau}{2\rho_{1}^{2}}\right)\right).
\label{eq:CCSR2}
\end{split}
\end{align}
 }
{
\begin{align}
E_{H}&\left[P(\xi_{2})P_{e}\vert\xi_{2}\right] \le \frac{c}{12}\left(\frac{\rho}{\sigma}\right)^{6}\left(1-\exp\left(-\frac{\tau}{2\rho^{2}}\right)\right)
\nonumber\\
&\left(\exp\left(-\frac{\tau}{2\rho^{2}}\right)\right)+\frac{c}{4}\left(\frac{\rho_{1}}{\sigma}\right)^{6}\left(1-\exp\left(-\frac{\tau}{2\rho_{1}^{2}}\right)\right)
\nonumber\\
&\left(\exp\left(-\frac{\tau}{2\rho_{1}^{2}}\right)\right).
\label{eq:CCSR2}
\end{align}
} 
Also, it can be shown that
\ifthenelse{\boolean{single}}
{
\begin{align}
\begin{split}
E_{H}&\left[P(\xi_{4})P_{e}\vert\xi_{4}\right]  =
  cE_{H}\left[P(\xi_{4})Q\left(\sqrt{\frac{g\Vert\mathcal{H}_{4}(\ell)\Vert^{2}\big\vert \xi_4}{N_{0}}}\right)\right]\\
\le & \frac{c}{12}\Big(\frac{\rho}{ \sigma}\Big)^{8}\Bigg(1-\exp\Big(-\frac{\tau}{2\rho^{2}}\Big)\Bigg)^{2}+ \frac{c}{4}\Big(\frac{\rho_{1}}{ \sigma}\Big)^{8}\Bigg(1-\exp\Big(-\frac{\tau}{2\rho_{1}^{2}}\Big)\Bigg)^{2},
\label{eq:CCSR3}
\end{split}
\end{align}
 }
{
\begin{align}
E_{H}&\left[P(\xi_{4})P_{e}\vert\xi_{4}\right]  =
  cE_{H}\left[P(\xi_{4})Q\left(\sqrt{\frac{g\Vert\mathcal{H}_{4}(\ell)\Vert^{2}\big\vert \xi_4}{N_{0}}}\right)\right]\nonumber\\
\le & \frac{c}{12}\Big(\frac{\rho}{ \sigma}\Big)^{8}\Bigg(1-\exp\Big(-\frac{\tau}{2\rho^{2}}\Big)\Bigg)^{2}+ \frac{c}{4}\Big(\frac{\rho_{1}}{ \sigma}\Big)^{8}\nonumber\\
&\Bigg(1-\exp\Big(-\frac{\tau}{2\rho_{1}^{2}}\Big)\Bigg)^{2},
\label{eq:CCSR3}
\end{align}
} 
where $\chi_{1} < \tau$, $\chi_{1s}\in \mathcal{R}$, $\chi_{2} < \tau$ and $\chi_{2s}\in \mathcal {R}$.
Now  using \eqref{eq:CCSR1},  \eqref{eq:CCSR2} and \eqref{eq:CCSR3} in \eqref{eq:CCSRJointQfunc}, we have
\ifthenelse{\boolean{single}}
{
\begin{align}
\begin{split}
P_{e_{2}} & \le\frac{c}{12}\Big(\frac{\rho}{\sigma}\Big)^{4}\Bigg(\exp\Big(-\frac{\tau}{2\rho^{2}}\Big)\Bigg)^{2}+\frac{c}{4}\Big(\frac{\rho_{1}}{\sigma}\Big)^{4}\Bigg(\exp\Big(-\frac{\tau}{2\rho_{1}^{2}}\Big)\Bigg)^{2}+\\
&\frac{c}{6}\Big(\frac{\rho}{\sigma}\Big)^{6} \Bigg(\exp\Big(-\frac{\tau}{2\rho^{2}}\Big)\Bigg)\Bigg(1-\exp\Big(-\frac{\tau}{2\rho^{2}}\Big)\Bigg)+\\
&\frac{c}{2}\Big(\frac{\rho_{1}}{\sigma}\Big)^{6}\Bigg(\exp\Big(-\frac{\tau}{2\rho_{1}^{2}}\Big)\Bigg) \Bigg(1-\exp\Big(-\frac{\tau}{2\rho_{1}^{2}}\Big)\Bigg)+\\
&\frac{c}{12}\Big(\frac{\rho}{ \sigma}\Big)^{8}\Bigg(1-\exp\Big(-\frac{\tau}{2\rho^{2}}\Big)\Bigg)^{2}
  +\frac{c}{4}\Big(\frac{\rho_{1}}{ \sigma}\Big)^{8}\Bigg(1-\exp\Big(-\frac{\tau}{2\rho_{1}^{2}}\Big)\Bigg)^{2}.
\label{eq:BERUBoundJointCCSR}
\end{split}
\end{align}
 }
{
\begin{align}
&P_{e_{2}}  \le\frac{c}{12}\Big(\frac{\rho}{\sigma}\Big)^{4}\Bigg(\exp\Big(-\frac{\tau}{2\rho^{2}}\Big)\Bigg)^{2}+\frac{c}{4}\Big(\frac{\rho_{1}}{\sigma}\Big)^{4}\nonumber\\
&\Bigg(\exp\Big(-\frac{\tau}{2\rho_{1}^{2}}\Big)\Bigg)^{2}+\frac{c}{6}\Big(\frac{\rho}{\sigma}\Big)^{6} \Bigg(\exp\Big(-\frac{\tau}{2\rho^{2}}\Big)\Bigg)\nonumber\\
&\Bigg(1-\exp\Big(-\frac{\tau}{2\rho^{2}}\Big)\Bigg)+
\frac{c}{2}\Big(\frac{\rho_{1}}{\sigma}\Big)^{6}\Bigg(\exp\Big(-\frac{\tau}{2\rho_{1}^{2}}\Big)\Bigg) \nonumber\\
&\Bigg(1-\exp\Big(-\frac{\tau}{2\rho_{1}^{2}}\Big)\!\Bigg)\!+\!\frac{c}{12}\Big(\frac{\rho}{ \sigma}\Big)^{8}\Bigg(1-\exp\Big(-\frac{\tau}{2\rho^{2}}\Big)\!\Bigg)^{2}
  \nonumber\\
&+\frac{c}{4}\Big(\frac{\rho_{1}}{ \sigma}\Big)^{8}\Bigg(1-\exp\Big(-\frac{\tau}{2\rho_{1}^{2}}\Big)\Bigg)^{2}.
\label{eq:BERUBoundJointCCSR}
\end{align} 
} 
The following proposition generalizes BER upper bound on joint detection for $J$ transmission rounds:
\begin{prop}
The upper bound on BER for joint decoding of $J$ transmission rounds under selective retransmissions is
\ifthenelse{\boolean{single}}
{
\begin{align}
P_{e_{J}}  =&\displaystyle \sum_{i=0}^{J}\Bigg\{\binom{J}{i} \frac{c}{12}   \left(\frac{\rho}{\sigma}\right)^{2(J+i)}\Bigg(\exp\left(-\frac{\tau}{2\rho^{2}}\right)\Bigg )^{(J-i)} \Bigg ( 1-\exp \left(-\frac{\tau}{2\rho^{2}}\right)\Bigg )^{i}  +\nonumber\\
&\binom{J}{i} \frac{c}{4}   \left(\frac{\rho_1}{\sigma}\right)^{2(J+i)}\Bigg(\exp\left(-\frac{\tau}{2\rho_1^{2}}\right)\Bigg )^{(J-i)}\Bigg ( 1-\exp \left(-\frac{\tau}{2\rho_1^{2}}\right)\Bigg )^{i}\Bigg\}. 
\end{align}
 }
{
\begin{align}
&P_{e_{J}}  =\displaystyle \sum_{i=0}^{J}\Bigg\{\binom{J}{i} \frac{c}{12}   \left(\frac{\rho}{\sigma}\right)^{2(J+i)}\Bigg(\exp\left(-\frac{\tau}{2\rho^{2}}\right)\Bigg )^{(J-i)} \nonumber\\
&\Bigg ( 1-\exp \left(-\frac{\tau}{2\rho^{2}}\right)\Bigg )^{i}  +\binom{J}{i} \frac{c}{4}   \left(\frac{\rho_1}{\sigma}\right)^{2(J+i)}\nonumber\\
&\Bigg(\exp\left(-\frac{\tau}{2\rho_1^{2}}\right)\Bigg )^{(J-i)}\Bigg ( 1-\exp \left(-\frac{\tau}{2\rho_1^{2}}\right)\Bigg )^{i}\Bigg\}. 
\end{align} 
} 
\end{prop}
\begin{proof}
For $J$ transmission rounds, there are $2^J$ possible joint detection channel vectors and we define $J+1$ events. The event $\xi_{i}$ consists of  $ \binom{J}{i}$ joint channel vectors. Each joint vector includes channel gains from the full transmissions and selective retransmissions. Note that when channel gain  of the $\ell$-th sub-carrier of the full transmission of each round has $\Vert H_{j} \Vert^2 \ge  \tau$, where $j=1,\hdots, J$,  the size of joint channel vector is $J$. The event $\xi_i$ occurs when $i$ sub-carrier realizations out of $J$ realizations of the $\ell$-th sub-carriers corresponding to the full transmissions have   $\Vert H_j{(\ell)}\Vert ^2<\tau$ in any order. Thus, there are $ \binom{J}{i}$ joint channel vector realizations out of $2^{J}$ possible outcomes which  have $i$ sub-carrier gains below threshold $\tau$ in any order for the joint detection of the $J$-th round. 
For example, for joint detection of 4 round of CCSR ($J=4$) and $i=2$, there are $2$ channel gains out of 4 for which channel-norm is lower than $\tau$ in any order. The possible $ \binom{4}{2}=6$ many joint channel gain vectors  belong to the event $\xi_{2}$ are
\ifthenelse{\boolean{single}}
{
\begin{align*}
         \xi_{2}\! =\!\{&\!\!\left[ H_1(\ell)  \;\; H_{1s}(\ell) \;\; H_2(\ell)\;\;  H_{2s} (\ell) \;\;   H_3(\ell)  \;\;  H_4(\ell)\right]^T\!\!\!,
         \left[ H_1(\ell)   \;\; H_{1s}(\ell) \;\; H_2(\ell)\;\;  H_3(\ell) \;\;  H_{3s}(\ell)    \;\;  H_4(\ell)\right]^T\!\!\!,\\
         &\!\!\left[ H_1(\ell)   \;\; H_{1s}(\ell) \;\; H_2(\ell)\;\;  H_3(\ell) \;\;  H_4(\ell)   \;\; H_{4s}(\ell) \right]^T\!\!\!,
         \left[ H_1(\ell)  \;\; H_2(\ell)\;\; H_{2s}(\ell) \;\;  H_3(\ell) \;\;  H_{3s} (\ell)  \;\;   H_4(\ell)\right]^T\!\!\!,\\
         &\!\!\left[ H_1(\ell)  \;\; H_2(\ell)\;\; H_{2s}(\ell)  \;\;  H_3(\ell)  \;\; H_4(\ell)\;\;    H_{4s}(\ell)  \right]^T\!\!\!,
        \left[ H_1(\ell)  \;\; H_2(\ell)\;\; H_3(\ell) \;\; H_{3s}(\ell)  \;\;  H_4(\ell) \;\;   H_{4s}(\ell) \right]^T \}.
\end{align*}
 }
{
\begin{align*}
         \xi_{2}\! =\!\{&\!\!\left[ H_1(\ell)  \;\; H_{1s}(\ell) \;\; H_2(\ell)\;\;  H_{2s} (\ell) \;\;   H_3(\ell)  \;\;  H_4(\ell)\right]^T\!\!\!,\\
         &\left[ H_1(\ell)   \;\; H_{1s}(\ell) \;\; H_2(\ell)\;\;  H_3(\ell) \;\;  H_{3s}(\ell)    \;\;  H_4(\ell)\right]^T\!\!\!,\\
         &\!\!\left[ H_1(\ell)   \;\; H_{1s}(\ell) \;\; H_2(\ell)\;\;  H_3(\ell) \;\;  H_4(\ell)   \;\; H_{4s}(\ell) \right]^T\!\!\!,\\
         &\left[ H_1(\ell)  \;\; H_2(\ell)\;\; H_{2s}(\ell) \;\;  H_3(\ell) \;\;  H_{3s} (\ell)  \;\;   H_4(\ell)\right]^T\!\!\!,\\
         &\!\!\left[ H_1(\ell)  \;\; H_2(\ell)\;\; H_{2s}(\ell)  \;\;  H_3(\ell)  \;\; H_4(\ell)\;\;    H_{4s}(\ell)  \right]^T\!\!\!,\\
        &\left[ H_1(\ell)  \;\; H_2(\ell)\;\; H_3(\ell) \;\; H_{3s}(\ell)  \;\;  H_4(\ell) \;\;   H_{4s}(\ell) \right]^T \}.
\end{align*}
}       
The elements of joint sub-carrier  gain vectors are independent with Gaussian distribution of zero-mean and variance $\sigma^2=\frac{1}{2}$. All the 6 joint channel vectors have equal impact on BER and are equivalent. Therefore, without loss of generality,  all 6 channel gain vectors of  event $\xi_{2}$ can be represented by a single joint channel vector 
\begin{align*}
\mathcal{H}_2=&\left[ H_1(\ell)  \;\;  H_2(\ell)\;\;    H_3(\ell)  \;\;  H_4(\ell)\;\;H_{1s}(\ell)\;\; H_{2s} (\ell) \;\; \right]^T.
\end{align*}
 The joint gain vector, which represents all joint channels of the event $\xi_{i}$  of the $J$-th round can be written as 
\ifthenelse{\boolean{single}}
{
\begin{align*}
\mathcal{H}_i=&\left [ \underbrace{H_{1}(\ell)\hdots H_{i}(\ell)}_{\Vert H_j(\ell) \Vert^2 < \tau} \; \underbrace{H_{i+1}(\ell)\; \hdots\;H_{J}(\ell)}_{\Vert H_j(\ell) \Vert^2 \ge \tau}\;\;H_{1s}(\ell)\hdots\;H_{is}(\ell)\right ]^T, 
\end{align*}
 }
{
\begin{align*}
\mathcal{H}_i=&\Bigg[ \underbrace{H_{1}(\ell)\hdots H_{i}(\ell)}_{\Vert H_j(\ell) \Vert^2 < \tau} \; \underbrace{H_{i+1}(\ell)\; \hdots\;H_{J}(\ell)}_{\Vert H_j(\ell) \Vert^2 \ge \tau}\;\;\nonumber\\
&H_{1s}(\ell)\hdots\;H_{is}(\ell)\Bigg ]^T, 
\end{align*}
} 
and  probability of occurring event $\xi_{i}$ is $P(\xi_{i})=p_1^{(J-i)}p_2^{i}\binom{J}{i}$,  where $P(\Vert H_j{(\ell)}\Vert ^2\ge\tau)=P(\xi)=p_1$ and $P(\Vert H_j{(\ell)}\Vert ^2<\tau)=P(\xi^c)=p_2$. Note that $i$ gains of $\mathcal{H}_i$ have $\Vert  H_{j}\Vert^2 <  \tau$, where $j=1, \hdots,\;i$ and  $J-i$ gains of joint channel $\mathcal{H}_i$ have $\Vert  H_{j}\Vert^2 \ge \tau$, where $j=i+1, \hdots,\;J$ . The event $\xi_0$ has all channel gains with $\Vert  H_{j}\Vert^2 \ge \tau$ and event $\xi_J$ has all channel gains with $\Vert  H_{j}\Vert^2 < \tau$. Furthermore,  all elements of vector $\mathcal{H}_i$ have chi-square distribution of order 2. 
The  probability of error of joint detection of the $J$-th round is
\begin{align}
P_{e_{J}}  =&E_{H}\bigg[ \displaystyle \sum_{i=0}^{J}P(\xi_{i})P_{e\vert\xi_{i}}\bigg]=\displaystyle \sum_{i=0}^{J}E_{H}\big[ P(\xi_{i})P_{e\vert\xi_{i}}\big].
\label{eq:CCSRMjoint}
\end{align}
Now we evaluate $E_{H}\bigg[ P(\xi_{i})P_{e\vert\xi_{i}}\bigg]$ as follows:
\ifthenelse{\boolean{single}}
{
\begin{align}
\begin{split}
E_{H}&\left[P(\xi_{i})P_{eJ}\vert\xi_{i}\right]  =  cE_{H}\left[P(\xi_{i})Q\left(\sqrt{\frac{g\Vert\mathcal{H}_{i}(\ell)\Vert^{2}\big \vert \xi_i}{N_{0}}}\right)\right]\\
\le & \binom{J}{i}\frac{c}{12} \prod_{k=1}^{J-i}\int_{\tau}^{\infty}\exp(-g\frac{x_{k}}{2N_{0}})f_{\chi_{k}}(x_{k})dx_{k}\prod_{k={J-i+1}}^{J}\int_{0}^{\tau}\exp(-g\frac{x_{k}}{2N_{0}})f_{\chi_{k}}(x_{k})dx_{k}. \\
&\prod_{k={1}}^{i}\int_{0}^{\infty}\exp(-g\frac{x_{ks}}{2N_{0}})f_{\chi_{ks}}(x_{ks})dx_{ks}+\binom{J}{i}
 \frac{c}{4} \prod_{k=1}^{J-i}\int_{\tau}^{\infty}\exp(-g\frac{4x_{k}}{3.2N_{0}})f_{\chi_{k}}(x_{k})dx_{k}.\\
&\prod_{k={J-i+1}}^{J}\int_{0}^{\tau}\exp(-g\frac{4x_{k}}{3.2N_{0}})f_{\chi_{k}}(x_{k})dx_{k}.
\prod_{k={1}}^{i}\int_{0}^{\infty}\exp(-g\frac{4x_{ks}}{3.2N_{0}})f_{\chi_{ks}}(x_{ks})dx_{ks}\\
=& \frac{c}{12}\binom{J}{i}\left(\frac{\rho}{\sigma}\right)^{2(J+i)}\Bigg(\exp\left(-\frac{\tau}{2\rho^{2}}\right)\Bigg )^{(J-i)} \Bigg ( 1-\exp \left(-\frac{\tau}{2\rho^{2}}\right)\Bigg )^{i}  +\\
	&\frac{c}{4}\binom{J}{i}    \left(\frac{\rho_1}{\sigma}\right)^{2(J+i)}\Bigg(\exp\left(-\frac{\tau}{2\rho_1^{2}}\right)\Bigg )^{(J-i)}\Bigg ( 1-\exp \left(-\frac{\tau}{2\rho_1^{2}}\right)\Bigg )^{i} . 
\label{eq:CCSRgen}
\end{split}
\end{align}
 }
{
\begin{align}
E_{H}&\left[P(\xi_{i})P_{eJ}\vert\xi_{i}\right]  =  cE_{H}\left[P(\xi_{i})Q\left(\!\!\sqrt{\frac{g\Vert\mathcal{H}_{i}(\ell)\Vert^{2}\big \vert \xi_i}{N_{0}}}\right)\right]\nonumber\\
\le &\binom{J}{i}\frac{c}{12} \prod_{k=1}^{J-i}\int_{\tau}^{\infty}\exp(-g\frac{x_{k}}{2N_{0}})f_{{\chi}_{k}}(x_{k})dx_{k}\prod_{k={J-i+1}}^{J}\nonumber\\
\int_{0}^{\tau}&\exp(-g\frac{x_{k}}{2N_{0}})f_{{\chi}_{k}}(x_{k})dx_{k}\prod_{k={1}}^{i}\int_{0}^{\infty}\exp(-g\frac{x_{ks}}{2N_{0}})\nonumber\\
f&_{{\chi}_{ks}}(x_{ks})dx_{ks}+ \binom{J}{i}
 \frac{c}{4} \prod_{k=1}^{J-i}\int_{\tau}^{\infty}\exp(-g\frac{4x_{k}}{3.2N_{0}})\nonumber\\
f&_{{\chi}_{k}}(x_{k})dx_{k}\prod_{k={J-i+1}}^{J}\int_{0}^{\tau}\exp(-g\frac{4x_{k}}{3.2N_{0}})f_{{\chi}_{k}}(x_{k})dx_{k}\nonumber\\
&
\prod_{k={1}}^{i}\int_{0}^{\infty}\exp(-g\frac{4x_{ks}}{3.2N_{0}})f_{{\chi}_{ks}}(x_{ks})dx_{ks}\nonumber\\
=& \frac{c}{12}\binom{J}{i}    \left(\frac{\rho}{\sigma}\right)^{2(J+i)}\Bigg(\exp\left(-\frac{\tau}{2\rho^{2}}\right)\Bigg )^{(J-i)} \nonumber\\
&\Bigg ( 1-\exp \left(-\frac{\tau}{2\rho^{2}}\right)\Bigg )^{i}  +\frac{c}{4}\binom{J}{i}    \left(\frac{\rho_1}{\sigma}\right)^{2(J+i)}\nonumber\\
	&\Bigg(\exp\left(-\frac{\tau}{2\rho_1^{2}}\right)\Bigg )^{(J-i)}\Bigg ( 1-\exp \left(-\frac{\tau}{2\rho_1^{2}}\right)\Bigg )^{i} . 
\label{eq:CCSRgen}
\end{align}
} 

Now by  substituting  $E_{H}\big[ P(\xi_{i})P_{e\vert\xi_{i}}\big]$ in \eqref{eq:CCSRMjoint}, we have
\ifthenelse{\boolean{single}}
{
\begin{align}
\begin{split}
P_{e_{J}}  =&\displaystyle \sum_{i=0}^{J}\binom{J}{i} \frac{c}{12}   \left(\frac{\rho}{\sigma}\right)^{2(J+i)}\Bigg(\exp\left(-\frac{\tau}{2\rho^{2}}\right)\Bigg )^{(J-i)} \Bigg ( 1-\exp \left(-\frac{\tau}{2\rho^{2}}\right)\Bigg )^{i}  +\\
&\binom{J}{i} \frac{c}{4}   \left(\frac{\rho_1}{\sigma}\right)^{2(J+i)}\Bigg(\exp\left(-\frac{\tau}{2\rho_1^{2}}\right)\Bigg )^{(J-i)}\Bigg ( 1-\exp \left(-\frac{\tau}{2\rho_1^{2}}\right)\Bigg )^{i}. 
\label{eq:CCSRThrm}
\end{split}
\end{align}
 }
{
\begin{align}
&P_{e_{J}}  =\displaystyle \sum_{i=0}^{J}\binom{J}{i} \frac{c}{12}   \left(\frac{\rho}{\sigma}\right)^{2(J+i)}\Bigg(\exp\left(-\frac{\tau}{2\rho^{2}}\right)\Bigg )^{(J-i)} \nonumber\\
&\Bigg ( 1-\exp \left(-\frac{\tau}{2\rho^{2}}\right)\Bigg )^{i}  +\binom{J}{i} \frac{c}{4}   \left(\frac{\rho_1}{\sigma}\right)^{2(J+i)}\nonumber\\
&\Bigg(\exp\left(-\frac{\tau}{2\rho_1^{2}}\right)\Bigg )^{(J-i)}\Bigg ( 1-\exp \left(-\frac{\tau}{2\rho_1^{2}}\right)\Bigg )^{i}. 
\label{eq:CCSRThrm}
\end{align}
} 

\end{proof}
 In order to compute probability for error  of the proposed CCSR method with $\mu$ (maximum allowed MAC retransmissions), we consider  \eqref{eq:CCSRThrm} with highest possible $J=\mu$, where $J=1,2, \hdots,\mu$.
\noindent

In next section, we present throughput analysis and optimization with
respect parameter $\tau$ for CCSR for $\mu$ MAC retransmissions (rounds).

\section{Throughput Analysis}
\label{sec:throughAnalysis}
 Now we present throughput analysis of the proposed CCSR method. In
throughput analysis, we consider non-truncated ARQ which has infinite
many retransmission rounds. There are $\mu$ retransmissions rounds
of a HARQ process at MAC layer in one ARQ round as depicted in Figure~\ref{fig:CCSRMethod}.
One retransmission round of HARQ process consists of a full transmission
of HARQ packet followed by a selective retransmission in PHY.
In practice, transceiver
pair continues retransmission rounds until error-free packet is received
or maximum number of retransmission rounds are reached. For throughput
analysis, we follow conventional definition of throughput $\eta$,
which is the ratio of error-free information bits received $k$ to
the total number of bits transmitted $n$ ( $\eta=\frac{k}{n}$).
Note that $P_{e_{J}}$ is the bit-error probability of the joint detection of the $J$-th round
of MAC transmission given in \eqref{eq:CCSRThrm}. Assuming that each
bit in the frame is independent, probability of receiving an error-free
packet of length $L_{f}$ with probability of bit-error $P_{e_{J}}$
is $p_{c_{J}}=(1-P_{e_{J}})^{L_{f}}$. The probability of receiving
a bad packet is $p_{\epsilon_{J}}=1-p_{c_{J}}$. As a direct consequence
of joint detection, probability of bit-error $P_{e_{1}}>P_{e_{2}}>\hdots>P_{e_{\mu}}$
and probability of receiving correct packet $p_{c_{1}}<p_{c_{2}}<\hdots<p_{c_{\mu}}$. 

One transmission round of HARQ layer consists of $k$ information
bits of the full transmission and $mk$ bits of selective retransmission,
where $m=p_{2}=P(\Vert H_1(\ell)\Vert^2 < \tau)$. Thus, there are $I=k(1+m)$ bits transmitted in
one transmission round of MAC layer to the receiver. As a result of
joint detection, $P_{e_{J}}<P_{e_{J-1}}$, $p_{c_{J}}>p_{c_{J-1}}$
and $p_{\epsilon_{J}}<p_{\epsilon_{J-1}}$. The probability that a packet
fails after two transmission is $p_{\epsilon_{1}}p_{\epsilon_{2}}$.
Note that if CRC failure occurs after $\mu$ transmissions at MAC
layer, receiver discards observations of $\mu$ transmissions and ARQ
layer initiates a new round of transmission of the failed packet. Thus probability
of CRC failure at the end of $\mu$ transmissions is $\alpha={\displaystyle \prod_{j=1}^{\mu}p_{\epsilon_{j}}=p_{\epsilon_{1}}.p_{\epsilon_{2}}\hdots p_{\epsilon_{\mu}}}$.
The probability of a packet to fail after $q$ transmissions with
joint detection of $\mu$ transmissions at MAC layer is 
\begin{align}
p_{\epsilon_{q}} & =\big(p_{\epsilon_{1}}.p_{\epsilon_{2}}\hdots p_{\epsilon_{\mu}}\big)^{\gamma}\prod_{j=1}^{J}p_{\epsilon_{J-j}}=\alpha^{\gamma}\prod_{j=1}^{J}p_{\epsilon_{J-j}}
\end{align}
where $\gamma=\lfloor{\frac{q-1}{\mu}}\rfloor$, $J=\left[(q-1)\mod\mu\right]+1$
and $p_{\epsilon_{0}}=1$. Note that $\gamma=0,1,\hdots,\infty$,
$J=1,2,\hdots,\mu$ represent transmission count
at MAC layer. Since there is joint detection of at most $\mu$ packets
and observations are discarded in the event of successful decoding
or failure of every $\mu$-th packet, the probability receiving error-free
packet in the event of $\mu+1$-th transmission of a packet is $p_{c_{1}}$.
The probability of successful decoding of $q$-th transmission of a failed
packet is $p_{c_{J}}$.

The number of bits transmitted in $q$ transmissions of a packet is
$k(1+m)q$. The average number of bits that transmitter transmits
$n_{\mu}$ for successful decoding of a packet in given channel condition
with $\mu$ transmission rounds at MAC layer in one ARQ round is stated as follows: 
\begin{prop}
The expected number of information bits under maximum number of $\mu$ rounds at MAC layer and non-truncated retransmissions at ARQ layer is sum of $\mu$ summation series as

\begin{align}
\begin{split}n_{\mu}= & \sum_{J=1}^{\mu}n_{J,\mu},
\end{split}
\label{eq:Proposition2}
\end{align}
where
\ifthenelse{\boolean{single}}
{
\begin{align}
{\displaystyle n_{J,\mu}} & =Jb_{J}+(J+\mu)b_{J}\alpha+(J+2\mu)b_{J}\alpha^{2}+(J+3\mu)b_{J}\alpha^{3}+\hdots,
\end{align}
}
{
\begin{align}
{\displaystyle n_{J,\mu}} & =Jb_{J}+(J+\mu)b_{J}\alpha+(J+2\mu)b_{J}\alpha^{2}+\nonumber\\
&(J+3\mu)b_{J}\alpha^{3}+\hdots,
\end{align}
} 
$\alpha={\displaystyle \prod_{j=1}^{\mu}p_{\epsilon_{j}}}$, 
$b_{J}=Ip_{\epsilon_{1}}p_{\epsilon_{2}}\hdots p_{\epsilon_{J-1}}p_{c_{J}}=k(1+m){\displaystyle \prod_{j=1}^{J}p_{\epsilon_{J-j}}}p_{c_{J}}$ and $p_{\epsilon_{0}}=1$.  
\end{prop}
\begin{proof}
The expected number of information bits transmitted to deliver error-free
$k$ information bits for the proposed CCSR method are \cite{ShuCostMiller}
\ifthenelse{\boolean{single}}
{
\begin{align}
\begin{split}n_{\mu}= & Ip_{c_{1}}+2Ip_{\epsilon_{1}}p_{c_{2}}+3Ip_{\epsilon_{1}}p_{\epsilon_{2}}p_{c_{3}}+\hdots+\mu Ip_{\epsilon_{1}}p_{\epsilon_{2}}p_{\epsilon_{3}}\hdots p_{\epsilon_{\mu-1}}p_{c_{\mu}}+ (\mu+1)Ip_{\epsilon_{1}}p_{\epsilon_{2}}p_{\epsilon_{3}}\\& \hdots p_{\epsilon_{\mu}}p_{c_{1}}
+(\mu+2)Ip_{\epsilon_{1}}^{2}p_{\epsilon_{2}}p_{\epsilon_{3}}\hdots p_{\epsilon_{\mu}}p_{c_{2}}+(\mu+3)Ip_{\epsilon_{1}}^{2}p_{\epsilon_{2}}^{2}p_{\epsilon_{3}}\hdots p_{\epsilon_{\mu}}p_{c_{3}}+ \hdots+\\&(\mu+\mu)Ip_{\epsilon_{1}}^{2}p_{\epsilon_{2}}^{2}p_{\epsilon_{3}}^{2}\hdots p_{\epsilon_{\mu-1}}^{2}p_{\epsilon_{\mu}}p_{c_{\mu}}+\hdots +(\gamma\mu+1)Ip_{\epsilon_{1}}^{\gamma}p_{\epsilon_{2}}^{\gamma}\hdots p_{\epsilon_{\mu}}^{\gamma}p_{c_{1}}+(\gamma\mu+2)Ip_{\epsilon_{1}}^{\gamma}p_{\epsilon_{2}}^{\gamma}\\&
\hdots p_{\epsilon_{\mu}}^{\gamma}p_{\epsilon_{1}}p_{c_{2}}+ (\gamma\mu+3)Ip_{\epsilon_{1}}^{\gamma}p_{\epsilon_{2}}^{\gamma}\hdots p_{\epsilon_{\mu}}^{\gamma}p_{\epsilon_{1}}p_{\epsilon_{2}}p_{c_{3}}+\hdots+(\gamma\mu+\mu)Ip_{\epsilon_{1}}^{\gamma}p_{\epsilon_{2}}^{\gamma}\hdots p_{\epsilon_{\mu}}^{\gamma}p_{\epsilon_{1}}p_{\epsilon_{2}}\\&\hdots p_{\epsilon_{\mu-1}}p_{c_{\mu}}+\hdots
\label{eq:prop1_1}
\end{split}
\end{align}

}
{
\begin{align}
&n_{\mu}=  Ip_{c_{1}}+2Ip_{\epsilon_{1}}p_{c_{2}}+3Ip_{\epsilon_{1}}p_{\epsilon_{2}}p_{c_{3}}+\hdots+\mu Ip_{\epsilon_{1}}p_{\epsilon_{2}}p_{\epsilon_{3}}\nonumber\\
&\hdots p_{\epsilon_{\mu-1}}p_{c_{\mu}}+ 
	   (\mu+1)Ip_{\epsilon_{1}}p_{\epsilon_{2}}p_{\epsilon_{3}} \hdots p_{\epsilon_{\mu}}p_{c_{1}}+
	(\mu+2)Ip_{\epsilon_{1}}^{2}\nonumber\\
	&p_{\epsilon_{2}}p_{\epsilon_{3}}\hdots p_{\epsilon_{\mu}}p_{c_{2}}+
	(\mu+3)Ip_{\epsilon_{1}}^{2}p_{\epsilon_{2}}^{2}p_{\epsilon_{3}}\hdots p_{\epsilon_{\mu}}p_{c_{3}}+ \hdots+\nonumber\\
	&(\mu+\mu)Ip_{\epsilon_{1}}^{2}p_{\epsilon_{2}}^{2}p_{\epsilon_{3}}^{2}\hdots p_{\epsilon_{\mu-1}}^{2}p_{c_{\mu}}+\hdots +(\gamma\mu+1)Ip_{\epsilon_{1}}^{\gamma}p_{\epsilon_{2}}^{\gamma}\nonumber\\
	&\hdots p_{\epsilon_{\mu}}^{\gamma}p_{c_{1}}+(\gamma\mu+2)Ip_{\epsilon_{1}}^{\gamma}p_{\epsilon_{2}}^{\gamma}\hdots p_{\epsilon_{\mu}}^{\gamma}p_{\epsilon_{1}}p_{c_{2}}+ (\gamma\mu+3)\nonumber\\
	&Ip_{\epsilon_{1}}^{\gamma}p_{\epsilon_{2}}^{\gamma}\hdots p_{\epsilon_{\mu}}^{\gamma}p_{\epsilon_{1}}p_{\epsilon_{2}}p_{c_{3}}+\hdots+(\gamma\mu+\mu)Ip_{\epsilon_{1}}^{\gamma}p_{\epsilon_{2}}^{\gamma}\hdots p_{\epsilon_{\mu}}^{\gamma}\nonumber\\
&	p_{\epsilon_{1}}p_{\epsilon_{2}}\hdots p_{\epsilon_{\mu-1}}p_{c_{\mu}}+\hdots
\label{eq:prop1_1}
\end{align}
} 
By rearranging \eqref{eq:prop1_1}, we have
\ifthenelse{\boolean{single}}
{
\begin{align}
\begin{split}n_{\mu}= & Ip_{c_{1}}\Big(1+(\mu+1)\alpha+(2\mu+1)\alpha^{2}+(3\mu+1)\alpha^{3}+\hdots\Big)+\\
 & Ip_{\epsilon_{1}}p_{c_{2}}\Big(2+(\mu+2)\alpha+(2\mu+2)\alpha^{2}+(3\mu+2)\alpha^{3}+\hdots\Big)+\hdots+\\
+ & Ip_{\epsilon_{1}}p_{\epsilon_{2}} s p_{\epsilon_{J-1}}p_{c_{J}}\Big(J+(\mu+J)\alpha+(2\mu+J)\alpha^{2}+(3\mu+J)\alpha^{3}+\hdots\Big)+\hdots+\\
+ & Ip_{\epsilon_{1}}p_{\epsilon_{2}} s p_{\epsilon_{\mu-1}}p_{c_{\mu}}\Big(\mu+(\mu+\mu)\alpha+(2\mu+\mu)\alpha^{2}+(3\mu+\mu)\alpha^{3}+\hdots\Big)
\end{split}
\label{eq:prop1_2}
\end{align}
}
{
\begin{align}
&n_{\mu}=  Ip_{c_{1}}\Big(1+(\mu+1)\alpha+(2\mu+1)\alpha^{2}+(3\mu+1)\alpha^{3}+\nonumber\\
&\hdots\Big)+ Ip_{\epsilon_{1}}p_{c_{2}}\Big(2+(\mu+2)\alpha+(2\mu+2)\alpha^{2}+
(3\mu+2)\nonumber\\
&\alpha^{3}+\hdots\Big)+\hdots+ Ip_{\epsilon_{1}}p_{\epsilon_{2}} \hdots p_{\epsilon_{J-1}}p_{c_{J}}\Big(J+(\mu+J)\alpha+\nonumber\\
&(2\mu+J)\alpha^{2}+(3\mu+J)\alpha^{3}+\hdots\Big)+\hdots+Ip_{\epsilon_{1}}p_{\epsilon_{2}} \hdots \nonumber\\
&p_{\epsilon_{\mu-1}}p_{c_{\mu}}\Big(\mu+(\mu+\mu)\alpha+(2\mu+\mu)\alpha^{2}+(3\mu+\mu)\alpha^{3}+\nonumber\\
&\hdots\Big)
\label{eq:prop1_2}
\end{align}
} 
\ifthenelse{\boolean{single}}
{
\begin{align}
n_{\mu}= & b_{1}\Big(1+(\mu+1)\alpha+(2\mu+1)\alpha^{2}+(3\mu+1)\alpha^{3}+\hdots\Big)+\nonumber\\
 & b_{2}\Big(2+(\mu+2)\alpha+(2\mu+2)\alpha^{2}+(3\mu+2)\alpha^{3}+\hdots\Big)+\nonumber\\
&\hdots+\\
 & b_{J}\Big(J+(\mu+J)\alpha+(2\mu+J)\alpha^{2}+(3\mu+J)\alpha^{3}+\hdots\Big)+\nonumber\\
&\hdots+\nonumber\\
+ & b_{\mu}\Big(\mu+(\mu+\mu)\alpha+(2\mu+\mu)\alpha^{2}+(3\mu+\mu)\alpha^{3}+\hdots\Big)
\label{eq:prop1_3}
\end{align}
}
{
\begin{align}
&n_{\mu}= 
 b_{1}\Big(1+(\mu+1)\alpha+(2\mu+1)\alpha^{2}+(3\mu+1)\alpha^{3}+\nonumber\\
 & \hdots\Big)+b_{2}\Big(2+(\mu+2)\alpha+(2\mu+2)\alpha^{2}+(3\mu+2)\alpha^{3}\nonumber\\
&+\hdots\Big)+\hdots+ b_{J}\Big(J+(\mu+J)\alpha+(2\mu+J)\alpha^{2}+\nonumber\\
 &(3\mu+J)\alpha^{3}+\hdots\Big)+\hdots+ b_{\mu}\Big(\mu+(\mu+\mu)\alpha+\nonumber\\
& (2\mu+\mu)\alpha^{2}+(3\mu+\mu)\alpha^{3}+\hdots\Big)
\label{eq:prop1_3}
\end{align}
} 

where $b_{1}=Ip_{c_{1}}$, $b_{2}=Ip_{\epsilon_{1}}p_{c_{2}}$ and
$b_{J}=Ip_{\epsilon_{1}}p_{\epsilon_{2}}\hdots p_{\epsilon_{J-1}}p_{c_{J}}=k(1+m){\displaystyle \prod_{j=1}^{J}p_{\epsilon_{J-j}}p_{c_{J}}}$.
Note that there are $\mu$ summation series in \eqref{eq:prop1_3} and $J$-th summation series in the above expression is
\ifthenelse{\boolean{single}}
{
\begin{align}
n_{J,\mu} & =Jb_{J}+(J+\mu)b_{J}\alpha+(J+2\mu)b_{J}\alpha^{2}+(J+3\mu)b_{J}\alpha^{3}+\hdots.
\end{align}
}
{
\begin{align}
n_{J,\mu}  =&Jb_{J}+(J+\mu)b_{J}\alpha+(J+2\mu)b_{J}\alpha^{2}+(J+3\mu)\nonumber\\
&b_{J}\alpha^{3}+\hdots.
\end{align}
} 
Therefor, expected number of information bits need to be transmitted
to deliver $k(1+m)$ information bits is 

\begin{align}
\begin{split}n_{\mu} & ={\displaystyle \sum_{J=1}^{\mu}n_{J,\mu}.}\end{split}
\end{align}
\end{proof}
Now we are ready to state proposition for throughput of CCSR method
with joint detection of $\mu$ packet at MAC layer. The following proposition
presents throughput of CCSR method:
\begin{prop}
Throughput of CCSR method under maximum number of $\mu$ rounds at MAC layer and non-truncated retransmissions at ARQ layer is
\begin{align}
\eta_{\mathcal{\mu}}= & \frac{(1-\alpha)^{2}}{(1+m){\displaystyle \sum_{J=1}^{\mathcal{\mu}}\prod_{j=1}^{J}p_{\epsilon_{J-j}}p_{c_{J}}\big(J+(\mathcal{\mathcal{\mu}}-J)\alpha}\big)},
\label{eq:PropThroughput}
\end{align}
where ${\displaystyle \alpha=\prod_{J=1}^{\mathcal{\mathcal{\mu}}}p_{\epsilon_{J}}}$
and $p_{\epsilon_{0}}=1$. 
\end{prop}
\begin{proof}
From proposition 2, average number of bits $n_{\mathcal{\mu}}$ required to deliver error-free
packet to the receiver consists of summation of $\mu$
terms. That is,
\begin{align}
n_{\mu}= \sum_{J=1}^{\mathcal{\mu}}n_{J,\mathcal{\mu}}= n_{1,\mu}+n_{2,\mu}+ \hdots+n_{\mu,\mu}
\end{align}
The $J$-th summation series is 
\ifthenelse{\boolean{single}}
{
\begin{align}
n_{J,\mu} & =b_{J}\Big(J+(\mu+J)\alpha+(2\mu+J)\alpha^{2}+(3\mu+J)\alpha^{3}+\hdots\Big)
\label{eq:Substituent1}
\end{align}
}
{
\begin{align}
n_{J,\mu}  =&b_{J}\Big(J+(\mu+J)\alpha+(2\mu+J)\alpha^{2}+(3\mu+J)\alpha^{3}\nonumber\\
&+\hdots\Big)
\label{eq:Substituent1}
\end{align}
} 
\ifthenelse{\boolean{single}}
{
\begin{align}
\begin{split}
n_{J,\mu}(\alpha-1) & =b_{J}\Big (-J-\mu\alpha-\mu\alpha^{2}-\mu\alpha^{3}-\hdots\Big )\\
 & =b_{J}\Big (-J-\mu\alpha(1+\alpha+\alpha^{2}+\alpha^{3}+\hdots)\Big )=b_{J}\Big (-J-\frac{\mu\alpha}{(1-\alpha)}\Big ).
\label{eq:Substite1}
\end{split}
\end{align}
}
{
\begin{align}
n_{J,\mu}(\alpha-1)  &=b_{J}\Big (-J-\mu\alpha-\mu\alpha^{2}-\mu\alpha^{3}-\hdots\Big )\nonumber\\
 & =b_{J}\Big (-J-\mu\alpha(1+\alpha+\alpha^{2}+\alpha^{3}+\hdots)\Big )\nonumber\\
&=b_{J}\Big (-J-\frac{\mu\alpha}{(1-\alpha)}\Big ).
\label{eq:Substite1}
\end{align}
}

Thus, 
\begin{align}
 n_{J,\mu} =b_{J}\frac{J+(\mu-J)\alpha}{(1-\alpha)^{2}}.
\label{eq:Substite2}
\end{align}
Substituting $b_J$ in \eqref{eq:Substite2}, we have 
\begin{align}
\begin{split}
n_{J,\mu} & =(1+m)\prod_{k=1}^{J}p_{\epsilon_{J-k}}p_{c_{J}}\frac{J+(\mu-J)\alpha}{(1-\alpha)^{2}}.
\end{split}
\end{align}
The average number $n_{\mu}$ of transmitted bits required to 
deliver single error free bit at receiver under  $\mu$ transmission rounds at MAC layer
of CCSR method is, 
\begin{align}
\begin{split}n_{\mu} & =\sum_{J=1}^{\mu}n_{J,\mu}
 =(1+m)\sum_{J=1}^{\mu}\prod_{j=1}^{J}p_{\epsilon_{J-j}}p_{c_{J}}\frac{J+(\mu-J)\alpha}{(1-\alpha)^{2}}
\end{split}
\label{eq:AvgBitsn}
\end{align}
Throughput of CCSR method is
\begin{align}
\begin{split}
\eta_{\mu} & =\frac{1}{n_{\mu}}
 =\displaystyle\frac{(1-\alpha)^{2}}{\displaystyle(1+m)\sum_{J=1}^{\mu}\prod_{j=1}^{J}p_{\epsilon_{J-j}}p_{c_{J}}\big(J+(\mu-J)\alpha\big)}
\end{split}
\label{eq:ThroughputExp}
\end{align}
\end{proof}
Note that $\eta_{\mu}$ of CCSR is function of parameter $\tau$
that controls the information to be transmitted during selective retransmission.
The parameter $\tau$ can be optimize to maximize throughput under
OFDM signaling. Next, we discuss search for optimal $\tau$ for the
proposed CCSR method to enhance throughput of an OFDM transceiver.
\section{Throughput Optimization}
\label{sec:OptThrough} 
In this section, we optimize throughput
of the proposed selective retransmission method at modulation layer. The amount
of information that a receiver requests to the transmitter 
in the event of a packet failure has direct impact on the throughput
of the transceiver. Most of the time, especially in high SNR regime, receiver
can recover from bit errors by receiving little more information and employing joint detection.
In selective retransmission at modulation level, threshold $\tau$
on channel norm of a sub-carrier is measure of channel quality.
By choosing proper threshold $\tau$, receiver can request minimum
information needed to recover from errors for the failed packet. The threshold
$\tau$ is function of SNR and modulations such as 4-QAM and 16-QAM.
It is clear from \eqref{eq:PropThroughput}
that throughput of CCSR method is  a function
of frame-error rate (FER). Furthermore, FER is not a linear or quadratic
function of SNR and parameter $\tau$. Now we write unconstrained
optimization problem for throughput $\eta$ with respect to parameter
$\tau$ as follows:
\begin{equation}
\mathbf{\tau_{o}}=\underset{\tau}{\text{arg  max}}\quad \eta= f(\tau,SNR).
\label{eq:etao}
\end{equation}
\begin{figure}[ht]
\ifthenelse{\boolean{single}}{\centering \includegraphics[width=11cm]{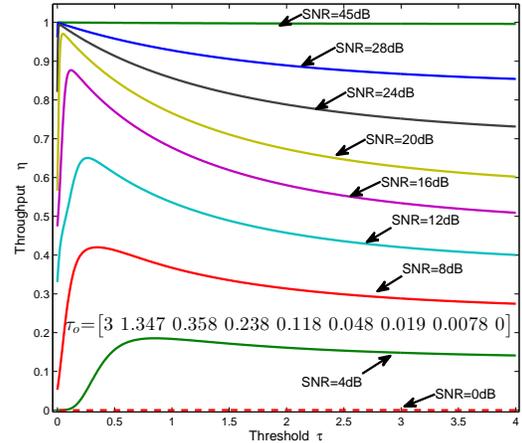}}
{\centering \includegraphics[width=8cm]{EtaVsTauFixingSNRCCSR}} 
\caption{Analytical throughput  \eqref{eq:PropThroughput} vs $\tau$ for  SNR operating points of CCSR method for $\mu=2$ transmission rounds.}
\label{fig:TauVsThroughputCCSR} 
\end{figure}
Since throughput $\eta$ is non-convex function in parameter $\tau$,
optimal $\tau$ that maximizes throughput $\eta$ for each SNR can
be computed off-line using exhaustive search. Thus, a table of optimal
threshold $\tau$ which maximizes throughput for  SNR operating points
can be generated using throughput expression in  \eqref{eq:PropThroughput}
 for CCSR method. Note that  throughput of the proposed CCSR method is function $\eta=f(\tau,\, SNR)$ given in  \eqref{eq:PropThroughput}.

In Section~\ref{Sec:simulation}, we maximize $\eta=f(\tau,\, SNR)$
with respect to parameter $\tau$. Note that parameter $\tau$ appears
in probability of frame error $p_{\epsilon}$, which is function of
probability of bit-error presented in \eqref{eq:CCSRThrm} Section~\ref{sec:perfomnceanlysis}. The optimal $\tau_o$ can be computed off-line from throughput lower bound for CCSR using  \eqref{eq:PropThroughput}. Based on channel condition, amount of information to be transmitted can be controlled using vector $\tau_o$. Figure~\ref{fig:TauVsThroughputCCSR} shows that optimal threshold $\tau_o$ for  SNR points which maximizes throughput CCSR method for $\mu=2$. In low SNR regime, throughput $\eta$ is more sensitive to threshold $\tau$ as compared to high SNR regime  due to the fact that in high SNR regime, very few errors occur during first transmission resulting into fewer retransmissions.


\section{Simulation}
\label{Sec:simulation} 
Now we present performance  of the proposed  CCSR  method in comparison with conventional Chase combining methods. In throughput performance, we consider optimized threshold $\tau_o$  which controls the amount of information in selective retransmission for OFDM systems. In simulation setup, we consider 4-QAM constellation and OFDM signaling with $N_{s}=512$ sub-carriers over 10-tap  Rayleigh fading frequency selective channel. Each complex OFDM channel realization has Gaussian
distribution with zero-mean and unit variance ($\sigma^2_h=1$). We
assume block fading channel in quasi-static fashion such that channel
remains highly correlated during transmission of one OFDM symbol. First,
we present comparison of  BER upper
bound  and BER from Monte Carlo simulation denoted by $BER_a$ and $BER_m$, respectively, for CCSR method. We also provide throughput
results of CCSR method in comparison with conventional Chase
combining method. We denote analytical and simulation throughput by $\eta_a$ and $\eta_m$, respectively. In order to maximize throughput, threshold $\tau$ on channel norm of OFDM
sub-carriers is optimized for each SNR point of CCSR protocol. We compute
threshold vector $\tau_o$ off-line to maximize throughput of CCSR method from the analytical throughput using  \eqref{eq:PropThroughput}. We also demonstrate that  our proposed CCSR method holds throughput gain  with FEC.  We consider half-rate LDPC code (648, 324) to evaluate efficacy of CCSR method as compared to conventional CC-HARQ. We denote CCSR method with FEC enabled by CCSR-HARQ in simulation results.  
 \begin{figure}[ht]
\ifthenelse{\boolean{single}}{\centering \includegraphics[width=11cm]{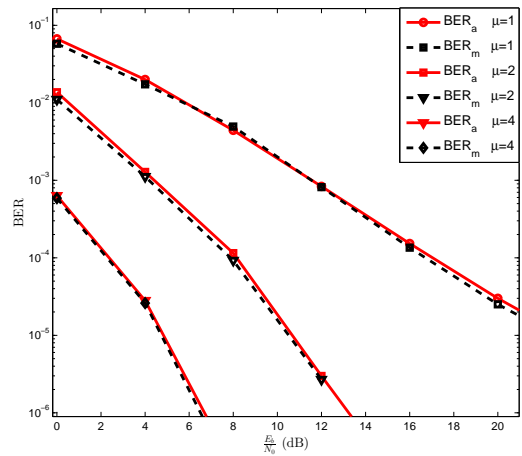}}
{\centering \includegraphics[width=8cm]{CCSRBERBound}} 
 \caption{Monte Carlo simulation Vs  BER upper bound \eqref{eq:CCSRThrm} of CCSR method for $\mu=1,2$ and $4$ transmission rounds with optimal threshold $\tau_o$.}
\label{fig:BERAnalVssimuCCSR} 
\end{figure}
\begin{figure}[h!]
\ifthenelse{\boolean{single}}{\centering \includegraphics[width=11cm]{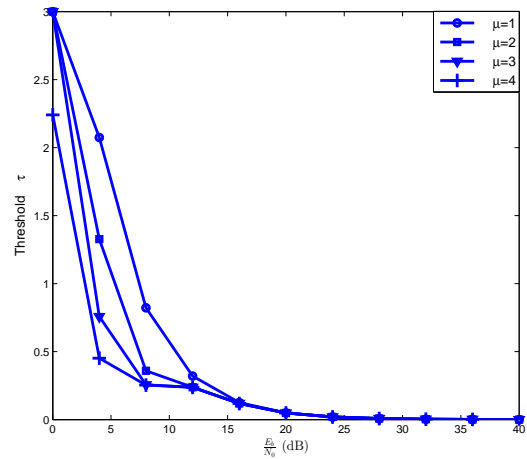}}
{\centering \includegraphics[width=8cm]{TauComparisonCCSRforDifferentJ}} 
 \caption{$\frac{E_b}{N_o}$ vs optimal threshold $\tau_o$ which maximizes throughput of CCSR-ARQ method for $\mu=1,2,3\;\text{and }\;4$ transmission rounds.} 
\label{fig:tauComparison1} 
\end{figure}

First, we present BER and throughput performance of CCSR method in comparison with conventional CC for $\mu=1,2$ and $4$ transmission rounds using optimal threshold $\tau_o$ without FEC. Note that each SNR point has an associated threshold $\tau_o$ which maximizes throughput at that very SNR point. For analytical BER and throughput performance, we use \eqref{eq:CCSRThrm} and \eqref{eq:PropThroughput}, respectively. We denote CCSR and conventional CC methods without FEC by CCSR-ARQ and CC-ARQ, respectively.
 Figure~\ref{fig:BERAnalVssimuCCSR} provides  BER comparison of Monte Carlo simulation and BER upper bound of CCSR-ARQ method for $\mu=1,2,\;\text{and }\;4$. Figure~\ref{fig:BERAnalVssimuCCSR} reveals that there is marginal gap between analytical and simulation bit-error rates $BER_a$ and  $BER_m$, respectively. 

 In Figure~\ref{fig:tauComparison1}, we present optimal threshold  $\tau_o$ as a function of SNR. As Figure~\ref{fig:tauComparison1} shows, at $\frac{E_b}{N_o} \ge 8$dB, optimal $\tau_o$ for $\mu=4$ converges to the $\tau_o$ for $\mu=3$. Similarly, at $\frac{E_b}{N_o} \ge 12$dB, optimal $\tau_o$ for $\mu=3$ and $\mu=4$ converge to the $\tau_o$ for $\mu=2$. At $\frac{E_b}{N_o} \ge 16$dB,  optimal $\tau_o$ for $\mu=2$, $\mu=3$ and $\mu=4$ converge to the $\tau_o$ for $\mu=1$. This is  due to the fact that as  $\frac{E_b}{N_o}$ increases, the probability of entering into next round of transmission (retransmission) at MAC  layer decreases. For example, at $\frac{E_b}{N_o} \ge 8$dB, probability that system reaches 4 rounds of transmission is very low and joint detection for $\mu=4$ converges to the joint detection  for $\mu=3$. We observe similar trend at $\frac{E_b}{N_o} =12,dB\;\text{and}\;16$dB for $\mu=3$ and $2$, respectively. 

\begin{figure}[t!]
\ifthenelse{\boolean{single}}{\centering \includegraphics[width=11cm]{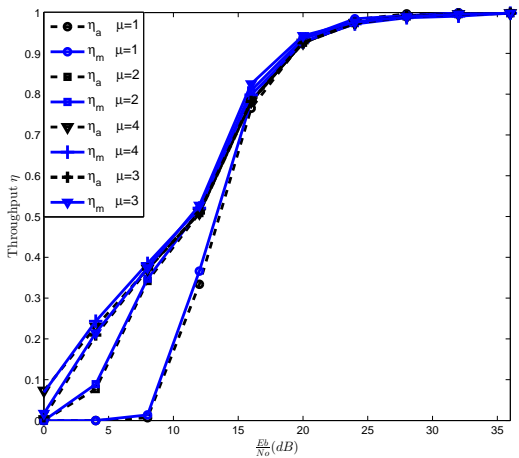}}
{\centering \includegraphics[width=8cm]{ThroughputBoundCCSR}} 
\caption{Monte Carlo simulations  Vs analytical throughput using \eqref{eq:PropThroughput} of CCSR-ARQ method for $\mu=1,2\;\text{and }\;4$ transmission rounds with optimal threshold $\tau_o$.}
\label{fig:througputCCSRptimizedVsConventionCC} 
\end{figure}

Throughput is the key performance metric of a communication system. Now we provide results for tight lower throughput bound and throughput comparison between CCSR-ARQ and CC-ARQ methods.
For throughput computation, we use standard definition of throughput of a communication system \cite{ShuLin} as $\dfrac{k}{n}.$
 In CCSR-ARQ method, threshold vector $\tau_o$ in Figure~\ref{fig:tauComparison1} is computed offline using \eqref{eq:etao} to maximize throughput $\eta$.
 Figure \ref{fig:througputCCSRptimizedVsConventionCC} presents comparison of analytical throughput using \eqref{eq:PropThroughput} with Monte Carlo throughput of CCSR-ARQ using optimal $\tau_o$ for $\mu=1,2\;\text{and}\;4$ transmission rounds. Marginal gap between analytical and simulation throughput demonstrate that   \eqref{eq:PropThroughput} provides tight throughput lower bound for CCSR-ARQ for $\mu=1,2,4$. Consistent with Figure~\ref{fig:tauComparison1}, at $\frac{E_b}{N_o} \ge 8$dB, throughput for $\mu=4$ converges to the throughput for $\mu=3$. Similarly, at $\frac{E_b}{N_o} \ge 12$dB, throughput for $\mu=4$ converge to the throughput for $\mu=2$. Also, at $\frac{E_b}{N_o} \ge 16$dB, throughput for $\mu=2$ and $\mu=4$ converge to the throughput for $\mu=1$.   

\begin{figure}[ht]
\ifthenelse{\boolean{single}}{\centering \includegraphics[width=11cm]{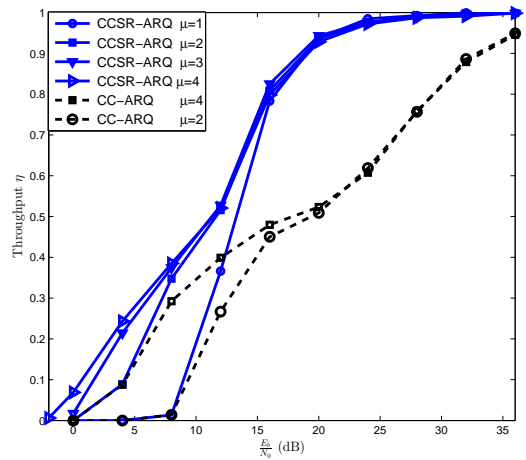}}
{\centering \includegraphics[width=8cm]{ThroughputCCSRvsCC}} 
\caption{ Throughput comparison of conventional Chase combining CC-ARQ $\mu=2,4$ transmission rounds and optimal CCSR-ARQ for $\mu=1,2,3,4$ transmission rounds using Monte Carlo simulations.}
\label{fig:ThroughputCCSRvsCCARQ} 
\end{figure}

  Now we provided throughput comparison between CCSR-ARQ and CC-ARQ methods for different number of transmission rounds. Figure \ref{fig:ThroughputCCSRvsCCARQ} demonstrates that CCSR-ARQ method achieves significant performance gain over conventional  CC-ARQ method for $\mu=1$, $2$, $3$ and $4$ transmission rounds. Note that throughput of CC-ARQ for $\mu=2$ is similar to that of CCSR for $\mu=1$ up to $\frac{E_b}{N_o}=8$dB. This is due to the fact that in low SNR regime, $\tau_o$ which maximizes throughput of the system has large value resulting selective retransmission of first round of CCSR into full packet retransmission. Thus, amount of information transmitted in first round of CCSR-ARQ equals amount of information of two rounds of CC-ARQ. For $\frac{E_b}{N_o}\ge 8$dB under $\mu=1$ for CCSR-ARQ, $\tau_o$ decreases as SNR increases and selective retransmission becomes more effective providing larger throughput gain over CC-ARQ. Thus, we  observe significant throughput gain of CCSR-ARQ method over CC-ARQ method in moderate and high SNR regime. We notice similar trend when we compare CCSR-ARQ for $\mu=2$ and CC-ARQ for $\mu=4$. It is important to note that in low SNR regime, $\tau_o$ converges to large values and throughput of CCSR-ARQ  converges to that of CC-ARQ. Figure \ref{fig:ThroughputCCSRvsCCARQ} also reveals that large transmission rounds are more effective in low SNR regime.
\begin{figure}[t!]
\ifthenelse{\boolean{single}}{\centering \includegraphics[width=11cm]{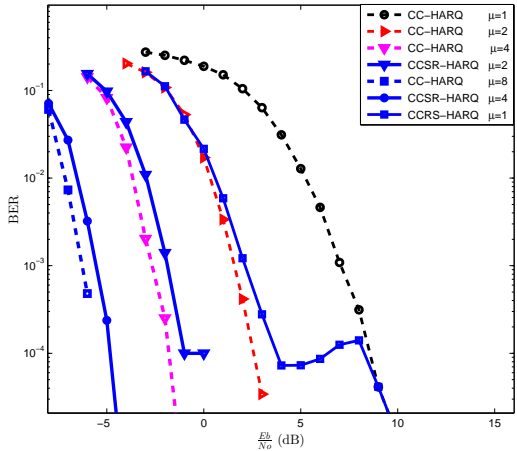}}
{\centering \includegraphics[width=8cm]{BERCCSRWithLDPCcopmarisonforAllJ}} 
\caption{BER comparison of CCSR-HARQ and conventional CC-HARQ methods for $\mu=1,2,4$ using half rate LDPC code (324, 648) with optimal $\tau_o$.}
\label{fig:HBERSCCwithFEC} 
\end{figure}

\begin{figure}[ht]
\ifthenelse{\boolean{single}}{\centering \includegraphics[width=11cm]{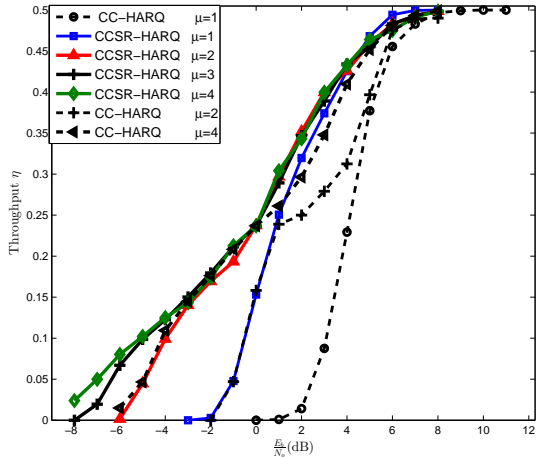}}
{\centering \includegraphics[width=8cm]{BERSCCthroughputcopmarisonOnlyLDPCNew}} 
\caption{Throughput comparison CC-HARQ for $\mu=1,2,4$ and CCSR-HARQ for $\mu=1,2,3,4$ transmission rounds with half rate LDPC code(324,648) using $\tau=\tau_o$.}
\label{fig:EtaSCCwithFEC} 
\end{figure}

Aforementioned results demonstrate significant throughput gain of CCSR-ARQ method over CC-ARQ method (without FEC). Now we present BER and throughput performance comparison between CCSR and CC methods with FEC. We denote CCSR and CC methods with FEC by CCSR-HARQ and CC-HARQ, respectively. Figure~\ref{fig:HBERSCCwithFEC} presents simulation results of  BER performance of conventional CC-HARQ and proposed CCSR-HARQ methods with half-rate LDPC (648,324) encoder under OFDM signaling. For BER simulation of CCSR-HARQ method, we use optimal threshold $\tau_o$ which maximizes throughput of CCSR-HARQ method for $\mu=1, 2,\;\text{and}\;4$.
 The threshold vectors for $\mu=1, 2,\;\text{and}\;4$ used in Figure~\ref{fig:HBERSCCwithFEC} and Figure \ref{fig:ThroughputCCSRvsCCARQ} are given in Table~\ref{Table1}.
\begin{table*}[t]
 \caption{SNR vs $\tau_o$ for $\mu=1, 2, 3$ and $4$} 
    \begin{tabular}{|r|l|l|l|l|l|l|l|l|l|l|l|l|l|l|l|}
        \hline
          $\frac{E_b}{N_o}  \big /  \mu $ & -8   &-6  &-5    &-3  &-2    &0  &1  & 2 &3  &4  &5  &6  &7  &8  \\ \hline
        1 &  3 & 3  & 3 & 2.2 & 1  & 0.595 & 0.43 & 0.301 & 0.16 & 0.090 & 0.030& 0.006 & 0 & 0 \\ \hline
        2 & 3 & 3 & 2.82  & 1.274 & 0.423  & 0.9 & 0.5 & 0.35 & 0.2 & 0.11 & 0.034 & 0.004 & 0 & 0 \\ \hline
        3  &3  &1.9&  1.741&  0.97 & 0.4   & 0 & 0.7 & 0.407 & 0.204 & 0.118 & 0.002 & 0 & 0 &0  \\ \hline
          4  &2.4   &1.338  &0.949   & 0.072 &0.023   &0  &0.63  &0.36  &0.235  &0.09  & 0.0038 &0.009 & 0.0005 &0  \\
        \hline
    \end{tabular}
		\label{Table1}
\end{table*}
Note that CCSR-HARQ and CC-HARQ methods are same for $\tau=0$. For very large value of threshold $\tau \rightarrow\infty$, BER performance of joint detection of CCSR-HARQ for $\mu$ transmission rounds is same as that of CC-HARQ for $2\mu$ rounds. 
In Figure~\ref{fig:HBERSCCwithFEC}, we use optimal threshold $\tau_o$ on selective retransmission in PHY which maximizes throughput $\eta$. In high SNR regime,  single round $\mu=1$ of CC-HARQ method achieves low probability of error resulting into low packet error rate and low probability of enetering into second round of transmission. In such scenario, optimal threshold $\tau_o\rightarrow 0$ and CCSR-HARQ method becomes CC-HARQ method. For $\frac{E_b}{N_o}> 4 \text{dB}$, $\tau_o\rightarrow 0$ and BER curve of CCSR-HARQ for $\mu=1$ merges with BER of CC-HARQ method.
  Similarly, we notice that under $\mu=2$ at  $\frac{E_b}{N_o}> -2 \text{dB}$, optimal threshold $\tau_o$, which maximizes throughput, is very small resulting into merging of BER for $\mu=2$ and $\mu=1$. In low SNR regime, BER and throughput of large transmission rounds is better than small transmission rounds. As channel condition improves, the CCSR-HARQ with large $\mu$ behaves similar to CCSR-HARQ method with small $\mu$.
Furthermore, at a given SNR, there are multiple values of $\tau$ which achieve similar throughput. The small threshold on channel norm $\tau$ has lower retransmission overhead and results into higher BER. As Figure~\ref{fig:HBERSCCwithFEC} reveals, BER marginally degrades when $\frac{E_b}{N_o}$ improves from 5 to 8 dB and still provides optimal throughput.

We present throughput of proposed CCSR-HARQ method in comparison with conventional CC-HARQ method in  Figure \ref{fig:EtaSCCwithFEC}. We provide Monte Carlo simulation results of throughput for CCSR-HARQ and CC-HARQ methods for $\mu=1,2,3,\;\text{and}\; 4$ transmission rounds using simulation setup of Figure~\ref{fig:HBERSCCwithFEC}. Throughput of CCSR-HARQ for $\mu=1$ is significantly higher than throughput of CC-HARQ for $\mu=1$, which demonstrates the efficacy of CCSR method. Also throughput of CCSR-HARQ for $\mu=1$ and CC-HARQ for $\mu=2$ are very close when $\frac{E_b}{N_o} \le 0$dB. However, for $\frac{E_b}{N_o} \ge 0$dB, CCSR-HARQ has higher throughput than CC-HARQ. For larger transmission rounds, CCSR-HARQ is more effective in low SNR regime as shown in Figure~\ref{fig:HBERSCCwithFEC}. Note that impact of selective retransmission in CCSR-HARQ method is not significant for $\mu=4$ in higher SNR regime due to the fact that CCSR-HARQ has low probability of CRC failure. Figure \ref{fig:EtaSCCwithFEC} also reveals that CCSR-HARQ for $\mu=4$ first converges to CCSR-HARQ for $\mu=3$ at $\frac{E_b}{N_o}\ge -6 $dB. Both curves continue in overlap fashion and converge to CCSR-HARQ for $\mu=2$ at $\frac{E_b}{N_o}\ge 0$dB and then $\mu=2,3,$ and $4$ follows CCSR-HARQ for $\mu=1$ at $\frac{E_b}{N_o}\ge 3 $dB similar to Figure \ref{fig:ThroughputCCSRvsCCARQ}. 

\vspace{-0.3cm}
\section{Conclusion}

\label{Sec:CONCLUSION} We presented bandwidth efficient cross-layer design under OFDM modulation
using selective retransmission sub-layer at PHY level. The throughput performance comparison demonstrated that the proposed CCSR method outperforms conventional Chase combing method. In performance analysis, we  provided tight BER upper bound  and throughput lower bound for the proposed CCSR method. The simulation results suggest that BER and throughput performances from Monte Carlo runs have marginal performance gap from that of BER upper bound and throughput lower bound, respectively.
In order to maximize throughput of the CCSR method, we  optimized threshold $\tau$ which controls amount of information to be retransmitted by embedded selective retransmission at PHY level. The simulation results also demonstrate significant
throughput gain of optimized selective retransmission method over
conventional Chase combining method with and without channel coding.
\vspace{0.3cm}

\ifCLASSOPTIONcaptionsoff \newpage{}\fi 
\bibliographystyle{IEEEtran}
\bibliography{CCWSAnalOptiThroughput}

\begin{thebibliography}{10}
\providecommand{\url}[1]{#1}
\csname url@samestyle\endcsname
\providecommand{\newblock}{\relax}
\providecommand{\bibinfo}[2]{#2}
\providecommand{\BIBentrySTDinterwordspacing}{\spaceskip=0pt\relax}
\providecommand{\BIBentryALTinterwordstretchfactor}{4}
\providecommand{\BIBentryALTinterwordspacing}{\spaceskip=\fontdimen2\font plus
\BIBentryALTinterwordstretchfactor\fontdimen3\font minus
  \fontdimen4\font\relax}
\providecommand{\BIBforeignlanguage}[2]{{%
\expandafter\ifx\csname l@#1\endcsname\relax
\typeout{** WARNING: IEEEtran.bst: No hyphenation pattern has been}%
\typeout{** loaded for the language `#1'. Using the pattern for}%
\typeout{** the default language instead.}%
\else
\language=\csname l@#1\endcsname
\fi
#2}}
\providecommand{\BIBdecl}{\relax}
\BIBdecl

\bibitem{LTEDOc}
``Long {T}erm {E}volution of the 3{GPP} {R}adio {T}echnology,''
  \emph{http:/www.3gpp.org}.

\bibitem{ShuLin}
S.~Lin and P.~S. Yu, ``A {H}ybrid {ARQ} {S}cheme with {P}arity {R}etransmission
  for {E}rror {C}ontrol of {S}atellite {C}hannels,'' \emph{IEEE Trans.
  Commun.}, vol.~30, no.~7, pp. 1701--1719, Jul. 1982.

\bibitem{HARQBroadBand}
S.~Y. Park and D.~Love, ``Hybrid {ARQ} {P}rotocol for {M}ulti-{A}ntenna
  {M}ulticasting {U}sing a {C}ommon {F}eedback {C}hannel,'' \emph{IEEE Trans.
  Commun.}, vol.~59, no.~6, pp. 1530--1542, Jun. 2011.

\bibitem{ChaseComb}
D.~Chase, ``{Code Combining - A Maximum-Likelihood Decoding Approach for
  Combining an Arbitrary Number of Noisy Packets},'' \emph{IEEE Tran. Commun.},
  vol.~33, pp. 385--393, May 1985.

\bibitem{ShenFitz}
C.~Shen and M.~P. Fitz, ``{Hybrid ARQ in Multiple-Antenna Slow Fading Channels:
  Performance Limits and Optimal Linear Dispersion Code Design},'' \emph{IEEE
  Trans. Inf. Theory}, vol.~57, no.~9, pp. 5863--5883, Sep. 2011.

\bibitem{ARQJindal}
P.~Wu and N.~Jindal, ``{Coding versus ARQ in Fading Channels: How Reliable
  Should the PHY Be?}'' \emph{IEEE Trans. Commun.}, vol.~59, no.~12, pp.
  3363--3374, December 2011.

\bibitem{JindalHARQAnalyRayFad}
------, ``{Performance of Hybrid-ARQ in Block-Fading Channels: A Fixed Outage
  Probability Analysis},'' \emph{IEEE Trans. on Comm.}, vol.~58, no.~4, pp.
  1129--1141, April 2010.

\bibitem{OptPowerLarson}
T.~C.~E. Larsson, ``{Optimal Power Allocation for Hybrid ARQ with Chase
  Combining in i.i.d. Rayleigh Fading Channels},'' \emph{IEEE Trans. on
  Commun.}, vol.~61, no.~5, pp. 1835 -- 1846, May 2013.

\bibitem{OutageLarson}
------, ``{Outage-Optimal Power Allocation for Hybrid ARQ with Incremental
  Redundancy},'' \emph{IEEE Trans. Wireless Commun.}, vol.~10, no.~7, pp.
  2069--2074, 2011.

\bibitem{AdapPowerLarson}
------, ``{Adaptive Power Allocation for HARQ with Chase Combining in
  Correlated Rayleigh Fading Channels},'' \emph{IEEE Commun. Lett.}, vol.~3,
  no.~2, pp. 169 -- 172, 2014.

\bibitem{EngAwareShami}
Dechene, J.~Dan, and A.~Shami, ``{Energy-Aware Resource Allocation Strategies
  for LTE Uplink with Synchronous HARQ Constraints},'' \emph{IEEE Trans. Mobile
  Computing}, vol.~13, no.~2, pp. 422--433, 2014.

\bibitem{GreenHARQ}
B.~Makki, G.~Amat, and T.~Eriksson, ``{Green Communication via Power-Optimized
  HARQ Protocols},'' \emph{IEEE Trans. Veh. Technol.}, vol.~63, no.~1, pp.
  161--177, Jan 2014.

\bibitem{ZiaDingGLOBECOM2008}
M.~Zia and Z.~Ding, ``Joint {ARQ} {R}eceiver {D}esign for {B}andwidth
  {E}fficient {MIMO} {S}ystems,'' in \emph{Proc. {IEEE} Global
  Telecommunication Conference ({GLOBECOM}'08)}, New Orleans, LO, Dec. 2008,
  pp. 1--5.

\bibitem{ZhiDingPiggy}
X.~Liang, C.~M. Zhao, and Z.~Ding, ``Piggyback {R}etransmissions over
  {W}ireless {MIMO} {C}hannels: {S}hared {H}ybrid-{ARQ} ({SHARQ}) for
  {B}andwidth {E}fficiency,'' \emph{IEEE Trans. Wireless Commun.}, vol.~12,
  no.~8, pp. 3770--3782, Aug. 2013.

\bibitem{ZiaDingHARQOSTBC}
M.~Zia and Z.~Ding, ``{Bandwidth Efficient Variable Rate HARQ Under Orthogonal
  Space-Time Block Codes},'' \emph{IEEE Trans. Signal Process.}, vol.~62,
  no.~3, pp. 3360--3390, Jul. 2014.

\bibitem{Jermey2007}
R.~Jermey and Z.~Ding, ``{C}hannel {E}stimation and {E}qualization {T}echniques
  in {D}ownsampled {ARQ} {S}ystems,'' \emph{IEEE Trans. Signal Process.},
  vol.~55, no.~5, pp. 2251--2262, May. 2007.

\bibitem{Alamouti}
S.~Alamouti, ``{A Simple Transmit Diversity Technique for Wireless
  Communications},'' \emph{IEEE J. Sel. Areas Commun.}, vol.~16, no.~8, pp.
  1451--1458, Oct. 1998.

\bibitem{ARQMinitor}
M.-W. Wu and P.-Y. Kam, ``{ARQ with Channel Gain Monitoring},'' \emph{IEEE
  Trans. Commun.}, vol.~60, no.~11, pp. 3342--3352, Nov. 2012.

\bibitem{ZiaSelHARQ}
M.~Zia, H.~Mahmood, A.~Ahmed, and N.~Saqib, ``{S}elective {HARQ} {T}ransceiver
  {D}esign for {OFDM} {S}ystem,'' \emph{IEEE Commun. Lett.}, vol.~17, no.~12,
  pp. 2229 -- 2232, Dec. 2013.

\bibitem{ETRIMIMOCondNum}
T.~Kiani, M.~Zia, N.~Abbas, and H.~Mahmood, ``{Bandwidth-Efficient Selective
  Retransmission for MIMO-OFDM },'' \emph{ETRI}, vol.~37, no.~1, pp. 66--67,
  2015.

\bibitem{IETOFDMPartial}
A.~Ahmed, M.~Zia, N.~Abbas, and H.~Mahmood, ``{Partial Automatic Repeat Request
  Transceiver for Bandwidth and Power Efficient Multiple-Input Multiple-Output
  Orthogonal Frequency Division Multiplexing Systems},'' \emph{IET
  Communications}, pp. 476--486, Mar. 2015.

\bibitem{SelecCCforOSTBC}
M.~Zia, ``{Selective Chase Combining for OSTB Coded MIMO-OFDM System},''
  \emph{International Journal Of Electronics and Communication}, 2015.

\bibitem{EricsonLTELLDesign}
A.~Larmo, M.~Lindstrom, M.~Mayer, G.~Pelletier, and J.~Tornser, ``{The LTE
  Link-Layer Design},'' \emph{IEEE Commun. Mag.}, vol.~47, no.~4, pp. 52 -- 59,
  2009.

\bibitem{LTEBookStafania}
S.~Stefania, I.~Toufik, and M.~Baker, \emph{{LTE - The UMTS Long Term
  Evolution: From Theory to Practice}}.\hskip 1em plus 0.5em minus 0.4em\relax
  Wiley, 2011.

\bibitem{Larsson}
E.~G. Larsson and P.~Stoica, \emph{Space-{T}ime {B}lock {C}oding for {W}ireless
  {C}ommunications}.\hskip 1em plus 0.5em minus 0.4em\relax Cambridge, 2003.

\bibitem{DgitalCommProakis}
J.~Proakis and M.~Salehi, \emph{Digital Communications}.\hskip 1em plus 0.5em
  minus 0.4em\relax McGraw-Hill, 2007.

\bibitem{StarkWoods}
H.~Stark and J.~W. Woods, \emph{{P}robability and {R}andom {P}rocesses with
  {A}pplications to {S}ignal {P}rocessing}.\hskip 1em plus 0.5em minus
  0.4em\relax Prentice Hall, 2001.

\bibitem{QFuncApprox}
M.~Chiani, D.~Dardari, and M.~K. Simon, ``{New Exponential Bounds and
  Approximations for the Computation of Error Probability in Fading
  Channels},'' \emph{IEEE Trans. Wireless Commun.}, vol.~2, no.~4, pp.
  840--845, July 2003.

\bibitem{ShuCostMiller}
S.~Lin, D.~J. Costello, and M.~J. Miller, ``{A}utomatic {R}epeat {R}equest
  {E}rror {C}ontrol {S}chemes,'' \emph{IEEE Commun. Mag.}, vol.~22, no.~12, pp.
  5--14, Dec. 1984.

\end{thebibliography}

\end{document}